\documentclass[12pt]{article}
\usepackage{amsfonts,amsmath, stmaryrd,dsfont}
\usepackage{graphicx}
\usepackage{natbib,subfig}
\usepackage{url, todonotes} 

\usepackage[ruled]{algorithm2e}
\SetKwInput{KwData}{Input}
\SetKwInput{KwResult}{Output}

\newcommand{\blind}{0}

\newcommand{\R}{\mathbb{R}} 
\renewcommand{\P}{\mathbb{P}}
\newcommand{\cA}{\mathcal{A}}
\newcommand{\cZ}{\mathcal{Z}}
\def\dblbr#1{\llbracket #1\rrbracket}
\newcommand{\1}{ \mathds{1} 
} 
\renewcommand{\l}{\ell}
\newcommand{\wh}[1]{{\widehat{#1}}}

\newtheorem{proposition}{Proposition}

\begin{document}

\def\spacingset#1{\renewcommand{\baselinestretch}%
{#1}\small\normalsize} \spacingset{1}



\if0\blind
{
  \title{\bf   Enhancing the Power of Gaussian Graphical Model Inference by Modeling the Graph Structure}
  \author{Valentin Kilian, Tabea Rebafka and Fanny Villers\thanks{
    The authors gratefully acknowledge \textit{please list all relevant funding sources in the unblinded version}} 
        }
        \date{LPSM, Sorbonne Université, Paris, France}
  \maketitle
} \fi

\if1\blind
{
  \bigskip
  \bigskip
  \bigskip
  \begin{center}
    {\LARGE\bf Title}
\end{center}
  \medskip
} \fi

\bigskip
\begin{abstract}
For the problem of inferring a  Gaussian graphical model (GGM), this  work explores the application of a recent   approach from the multiple testing literature for graph inference.  
The main idea of the method by \cite{RRV22}   is to model the data by a latent variable model, the so-called noisy stochastic block model (NSBM),   and then use the associated  $\ell$-values to infer the graph. The inferred graph controls the false discovery rate, that means that the proportion of falsely declared edges does not exceed a user-defined nominal level.
Here it is shown that any test statistic from the GGM literature can be used as input for the NSBM approach to perform  GGM inference. 
To make the approach feasible in practice, a new, computationally   efficient inference algorithm for the NSBM is developed relying on a greedy approach to maximize the integrated complete-data likelihood. 
Then an extensive numerical study illustrates that the NSBM approach outperforms the state of the art for any of the here considered GGM-test statistics. In particular in sparse settings and on real datasets a significant gain in power is observed.

 \end{abstract}

\noindent%
\vfill

\spacingset{1.75} 

\section{Introduction}\label{sec:intro}
Graphical modeling offers a convenient way to describe complex conditional dependence structures between variables using a graph. 
This is of great interest in  a large variety of domains of application. 
One of the main question is to attempt to discover the links between 
variables from a set of repeating observations of these variables. Gaussian Graphical Models (GGM) provide a suitable framework in this context.
The goal of our work is to propose an approach for GGM inference that 
is both powerful to detect significant links and allows a control of the error rate. More precisely, we aim at providing   a control of the false discovery rate (FDR),  that is we want to guarantee that  the proportion of falsely declared dependencies between variables does not exceed a user-defined nominal level.  
The approach that we propose here is based on a method  for graph inference  developed in   a general statistical framework by \cite{RRV22}. We  make this approach usable in practice for GGM inference by developing a new computationally efficient estimation algorithm for the so-called
 noisy stochastic block model and show that the approach detects more dependencies among variables than    state-of-the-art methods.

\subsection{Inference in Gaussian graphical models}
We consider the setting, where  observations are   independent realizations  of a multivariate  Gaussian distribution. Its precision matrix, that is the inverse of the covariance matrix,    provides information on the conditional correlations of the variables. Notably, two variables are conditionally independent given all other variables if and only if the associated entry of the precision matrix is zero. Now, GGM inference amounts to recover the non zero entries of the precision matrix from the realizations of the Gaussian distribution. The problem can also be viewed as a graph inference problem, where the edges of a binary graph indicate the non zero entries of the precision matrix and this graph is to be inferred from the Gaussian observations.

In the literature there are mainly three different approaches for  GGM inference. 
First, some methods aim at estimating the precision matrix, as done in graphical Lasso \citep{Banerjee08, Friedman} with a penalized  maximum likelihood approach, or in \citep{Clime, Transclime} with a constrained $\ell_1$ minimization method. 
Second,  there are methods  based on neighborhood estimation, as the  Lasso procedure by \cite{MB}, that performs a  conditional regression for every variable against   all other variables. 
  Third,     GGM inference can be done using  multiple-testing procedures based on appropriate estimators of the precision matrix.  
When the sample size is large and the sample covariance matrix  invertible, there are
 natural test statistics for the entries of the precision matrix  \citep{DrtonPerlman2007}. 
However, in the high-dimensional setting,  the sample covariance matrix is typically nonsigular, and there have been several propositions to circumvent this issue. The use of bagging or shrinkage provides more reliable estimates of the covariance matrix or its inverse \citep{SSa, SSb}.
Alternatively, one can  test  the first-order partial correlation, that is,  the conditional correlation of two variables given only another third one \citep{WB}.
More recently, new estimators of  the precision matrix were obtained by improving the original Lasso-regularized estimators 
 \citep{Liu2013, Ren2015, Jankova2018}. 
As these estimates  are asymptotically normal when the precision matrix is sparse, they can serve to infer the graph while controlling the FDR  \citep{Liu2013}. 

\subsection{Graph inference by a multiple testing procedure}
 \cite{RRV22} propose a  multiple testing procedure for graph inference in  the general statistical framework, where 
 a real-valued graph is observed resulting from a perturbation of some underlying binary latent graph, which is to be recovered.
We will show that this approach can be used  for GGM inference, by using 
test statistics from the GGM literature  as input to this method.

 The  procedure by  \cite{RRV22}  consists in modeling   the data by a suitable latent variable model and then  using the
 associated structured $\ell$-values, that is,  the posterior probabilities of the presence of an edge,  to infer a graph such that  the FDR is controlled.  This method is  more powerful than other existing multiple testing procedures, since the $\ell$-value of an entry $(i,j)$ of the graph depends not only on the observed value for $(i,j)$, but also on the learned statistical graph model.

 Data are modeled by the so-called noisy stochastic block model (NSBM), where the underlying binary graph is a classical stochastic block model,  to which the following blurring mechanism is applied:   absent edges are replaced with  pure random noise, and in place of present edges, some effect or signal is observed. 
 For the testing procedure the NSBM has to be fitted to the data, but inference is involved as the model is quite complex. In  \cite{RRV22}  a variational EM-algorithm is proposed, 
but the implemented algorithm is rather slow and not scalable to large graphs. Therefore, it is not  convenient  for GGM inference, where graphs can be huge. In this work, a new inference algorithm is developed according to the approach proposed by \cite{Come15}. 
While the variational EM-algorithm provides an approximation of the maximum likelihood estimate, the aforementioned  method is a Bayesian approach and relies on  a new objective function, the  so-called exact integrated complete-data likelihood (ICL). Its maximization is   a discrete optimization problem, that can be solved by a greedy algorithm. An important advantage of the  algorithm is that it  automatically selects the best model, that is, the best number of blocks in the NSBM.

\subsection{Contributions and organization of the paper}
 	 
In Section \ref{sec:NSBM},  the definition of the NSBM  and the associated multiple testing procedure for graph inference are recalled. 
In Section \ref{sec:ICL} a fast greedy algorithm  is developed to maximize the ICL criterion in the NSBM. 
 Then,   Section \ref{sec:GGM} shows   how to use the NSBM for GGM inference with control of the FDR.  Namely,   any known test statistic from the GGM literature can be used as an input for the NSBM approach.
 Finally, Section \ref{sec:Simus} provides a numerical study to assess the performance of our NSBM approach for GGM inference and a comparison to the state of the art on both simulated data and a real dataset.

\section{Model and graph inference procedure}
\label{sec:NSBM}
In this section we present notations, the model and the graph inference procedure.

\subsection{Model definition}
  In the noisy stochastic block model (NSBM) the observation matrix $X = (X_{i,j})_{i,j}\in\R^{p\times p}$ is  a perturbed version of some underlying binary latent graph represented by its adjacency matrix $A=(A_{i,j})_{i,j}\in\{0,1\}^{p\times p}$, which is modeled by a stochastic block model (SBM). Here only undirected graphs  without loops are considered.  Denote $\cA=\{(i,j): 1\leq i < j\leq p\}$ the set of all possible undirected edges among the   $p$ nodes of the graph. 
 
The nodes are partitioned into $Q$ blocks by a vector $Z=(Z_1,\dots,Z_p)$ which is composed of independent and identically distributed discrete latent variables $Z_i$   with probabilities $\pi_q=\P(Z_1=q)$ for $q\in\dblbr Q$
	for some parameter $\pi=(\pi_1,\dots,\pi_Q)\in[0,1]^Q$ such that $\sum_{q\in\dblbr Q} \pi_q=1$.  
	The entries $A_{i,j}$ of the latent graph $A$ are  independent conditionally on $Z$ with Bernoulli distribution, that is, 
	$$(A_{i,j})_{(i,j)\in \cA}\:|\:Z \sim \bigotimes_{(i,j)\in \cA} \mathcal{B}(w_{Z_i,Z_j}),$$
	for some symmetric parameter matrix $w=(w_{q,l})_{q,l}\in[0,1]^{Q\times Q}$. 
 The blurring mechanism is the following. If there is no edge between nodes $i$ and $j$, i.e. $A_{i,j}=0$, than $X_{i,j}$ is  standard Gaussian noise,  otherwise (when $A_{i,j}=1$) some signal is observed and the signal strength is modulated by the latent variables $Z_i$ and $Z_j$. That is,
 	$$
	(X_{i,j})_{(i,j)\in \cA}\:|\:Z,A \sim
	\bigotimes_{(i,j)\in \cA}   (1-A_{i,j}) \mathcal{N}(0,1) + A_{i,j}  \mathcal N(\mu_{Z_i,Z_j},\sigma^2),
	$$
   	for  an unknown symmetric parameter matrix $\mu=(\mu_{q,l})_{q,l} \in \R^{Q\times Q}$ and known variance  $\sigma^2>0$.   
 The  unknown   model  parameter in   the NSBM is  denoted by $\theta=(\pi,w,\mu)$ 
 and the  distribution and density of $(X,A,Z)$  are denoted by $\P_{\theta}$ and $p_\theta$, respectively.
In the Supplementary Material an extension of the NSBM is presented, where the variance of the Gaussian under the alternative is unknown  and assumed to  depend on the latent variables.

\subsection{Graph inference via multiple testing}
Graph inference in the NSBM can be viewed as the  multiple testing problem, where  for all pairs $(i, j)\in\cA$ we consider the hypotheses of the absence or presence of an edge  between nodes $i$ and $j$ in the latent graph $A$. That is, we test
$H_{0,i,j}: A_{i,j} = 0$ against  $H_{1,i,j}: A_{i,j} \neq 0$.

 \cite{RRV22} propose a powerful testing procedure based on     the posterior probabilities  that the null hypothesis is true, also called $\l$-values. They are given by
\begin{align*}
	\l_{i,j}(X,Z,\theta)&=\P_\theta(A_{i,j}=0\mid X,Z)= \frac{ (1-w_{Z_i,Z_j}) f_{\mathcal N(0,1)}(X_{i,j}) }{(1-w_{Z_i,Z_j})f_{\mathcal N(0,1)}(X_{i,j}) +w_{Z_i,Z_j} f_{\mathcal N(\mu_{Z_i,Z_j},\sigma^2)}(X_{i,j}) },
\end{align*}
where  $f_{\mathcal{N}(m,s^2)}$  denotes the Gaussian density with mean $m$ and variance $s^2$.  The procedure consists in rejecting $H_{0,i,j}$ if  the $\l$-value is sufficiently small, that is
$\l_{i,j}(X, Z; \theta) \le t$ for some threshold $t$, which can be  chosen such that the FDR is controlled. For more details we refer the reader to \cite{RRV22}.

The point of the $\l$-value $\l_{i,j}$ for  the node pair $(i,j)$  is that its value does not only depend  on the observation $X_{i,j}$, but  on the entire NSBM
through the model parameter  $\theta$ and the node clustering $Z$. This results in an important increase of the power compared to the classical Benjamini-Hochberg procedure. \cite{RRV22} also provide results on the optimality of this $\l$-values procedure in terms of FDR and the power. 
The approach is typically  used with estimated model parameters $\hat\theta$ and estimated latent variables $\hat Z$. Therefore, an efficient inference algorithm for the NSBM is required.  
  
 \section{New inference algorithm  in the NSBM }
\label{sec:ICL}

While \cite{RRV22} propose a variational EM-algorithm limited to   networks of very moderate size, in this section a much more efficient and more convenient inference algorithm for the NSBM is developed. The approach relies on an integrated likelihood criterion and comes with a fast greedy algorithm that computes the best clustering $Z$, estimates the model parameter $\theta$ and 
automatically performs model selection, that is the choice of the best number of blocks in the SBM.
 
 \subsection{ICL criterion in the NSBM}
The integrated complete-data log-likelihood (ICL) is defined in a Bayesian framework, where $\boldsymbol\pi_\theta$ denotes a prior on the parameter $\theta$. In the NSBM  the ICL is given by
\begin{align*}
\mathrm{ICL}^\mathrm{NSBM}(X, A, Z)&=	\log p(X, A, Z) = \log \left(\int_{\Theta} p_\theta(X, A, Z) \boldsymbol\pi_\theta(\theta)d\theta \right).
\end{align*}
Notably, the criterion only depends on the observations and the latent variables. 
The value of $Z$ that maximize the ICL can be used to derive an estimator of $\theta$, which in turn will be used to compute the $\ell$-values for the final test procedure.
 
 To compute the maximum it is convenient to have an analytic expression of $\mathrm{ICL}^\mathrm{NSBM}$. 
This is obtained by considering   a factorized prior    of the form $\boldsymbol\pi_\theta(\pi, w, \mu) = \boldsymbol\pi_\pi(\pi)\boldsymbol \pi_w(w) \boldsymbol\pi_\mu(\mu)$ and by using  conjugate priors  for the different parameters. Here we choose   the Dirichlet   distribution for the cluster proportions $\pi$, Beta   distributions for the edge probabilities $w$ and Gaussian   distributions for the means $\mu$. More precisely,
for convenient hyperparameters $n_0, \eta_0, \zeta_0, \rho_0, \tau_0$, we consider
\begin{align}
\boldsymbol \pi_\pi(\pi) &= \mathrm{Dir}(n_0, ..., n_0),\qquad
\boldsymbol \pi_w(w) = \bigotimes_{ q \leq l }  \boldsymbol\pi_{w_{q,l}}(w_{q,l}) = \bigotimes_{  q \leq l } \mathrm{Beta}(\eta_0, \zeta_0),\label{def:prior}\\
\boldsymbol  \pi_\mu(\mu) &= \bigotimes_{ q \leq l }  \boldsymbol\pi_{\mu_{q,l}}(\mu_{q,l}) = \bigotimes_{  q \leq l }\mathcal N(\rho_0, \tau^2_0).\nonumber
\end{align} 

To state the formulae of $\mathrm{ICL}^\mathrm{NSBM}$, we introduce some further notations and useful count statistics. 
Let $N= {p(p-1)}/2$ be the maximal number of possible edges,  
	$M^A=\sum_{(i,j)\in\mathcal A}A_{i,j}$   the total number of edges
	and $N_Q=Q(Q+1)/2$ the number of pairs of blocks.
Let $Z_{i,q} = \1\{Z_i=q\}$ indicate  whether node $i$ is assigned  to block $q$ or not, and $\cZ_{q,l}^{i,j}=\1\{Z_{i,q}Z_{j,l} +Z_{i,l}Z_{j,q}>0\}$   whether nodes $i$ and $j$ belong to blocks  $q$ and $l$. Then,  $n_q=\sum_{i\in\dblbr p} Z_{i,q}$  is the number of nodes associated with block $q$, and $n_{q,l}=\sum_{(i,j) \in \cA} \cZ_{q,l}^{i,j}A_{i,j}$ (resp. 
   $\bar n_{q,l}=\sum_{(i,j) \in \cA} \cZ_{q,l}^{i,j}(1-A_{i,j})$) is the number of present (resp. absent) edges  between nodes in $q$ and $l$. Finally,
   $I_{q,l} = \{(i,j) \in \cA, A_{i,j} = 1, \cZ_{q,l}^{i,j} = 1\}$ is the set of present edges between nodes belonging to blocks  $q$ and $l$ (such that $|I_{q,l}|=n_{q,l}$),
and  
$S_{q,l}$ 
denotes the empirical variance associated with $\{X_{i,j}, (i, j)\in I_{q,l}\}$.

\begin{proposition} 	\label{ICL}
    In the NSBM with the priors defined in \eqref{def:prior}, the ICL  is given by
	$$\mathrm{ICL}^\mathrm{NSBM}(X, A, Z) =    \mathrm{ICL^{SBM}}(A,Z) +    \mathrm{ICL^{noise}}(X,A,Z),$$
	where $\mathrm{ICL^{SBM}}$ is the ICL in the SBM given by
	\begin{align} 
	\mathrm{ICL^{SBM}}(A,Z) &=  \log p(Z)+ \log p(A|Z)\notag\\
	& = \log \frac{C(n_0+n_1,\dots,n_0+n_Q)}{C(n_0,...,n_0)} + \sum_{  q \le l}\log \frac{B(\eta_{0} + n_{q,l},\xi_{0}+\bar n_{q,l})}{B(\eta_0,\xi_0)},\label{ICLSBM} 
	\end{align} 
where  $C(x_1,\dots,x_Q) = \frac{\prod_q \Gamma(x_q)}{\Gamma(\sum_q x_q)}$, $\Gamma(.)$ is the Gamma function and  $B(a,b) = \frac{\Gamma(a) \Gamma(b)}{\Gamma(a+b)}$   the Beta function. Furthermore, $\mathrm{ICL^{noise}}(X,A,Z) = \log p(X | A, Z) $ is the part of the ICL criterion associated with  the noise layer in the NSBM given by
	\begin{align} 
	&\mathrm{ICL^{noise}}(X,A,Z) \label{ICLNSBM}
		 = - \frac{N}{2} \log(2\pi) - \frac12\sum_{(i,j) \in \mathcal{A}} (1-A_{i,j}) X_{i,j}^2
		 - \left(M^A-N_Q\right)\log\sigma
		 \\
		&\quad - \frac{1}{2} \sum_{q \le l }\left\{\log\left(\sigma^2+ \tau_0^2n_{q,l} \right) 
		+   \frac{\tau_0^2n_{q,l}^2 }{\tau_0^2\sigma^2 n_{q,l}  +  \sigma^4} S_{q,l}
		 +  \frac{1}{\sigma^2+\tau_0^2 n_{q,l}} \sum_{(i, j) \in I_{q,l}} (X_{i,j} - \rho_0 )^2 
\right\}. \notag 
	\end{align}
\end{proposition}

 \subsection{Greedy NSBM inference algorithm}
The  maximization of the ICL is a discrete optimization problem and  we propose  a greedy  algorithm for this task, which is in line with the algorithm developed by \cite{Come15}. The algorithm successively picks  a node and searches its best block assignment in terms of the ICL criterion while keeping all other node assignments fixed.
  This amounts to evaluate the variation of the ICL   denoted by $\Delta_{g\rightarrow h}$ due to swapping  the selected node $i^*$ with current block membership $g$ to block $h$, where $\Delta_{g\rightarrow h}$   is 
  defined by
$$\Delta_{g\rightarrow h} = \mathrm{ICL^{NSBM}}(X,Z^\mathrm{swap}, A^\mathrm{swap}) - \mathrm{ICL^{NSBM}}(X,Z, A),$$
where $Z=(Z_1, \ldots, Z_p)$ denotes the current partition of the   nodes with $Z_{i^*}=g$ 
and $Z^\mathrm{swap}$ the block assignments  after   the swap, that is, $Z^\mathrm{swap}_{i^*}=h$ and $Z^\mathrm{swap}_{i }=Z_i$ for  all $i\neq i^*$.
Likewise,   $A$ denotes the current  adjacency matrix   and  $A^\mathrm{swap}$ the one after the swap obtained by thresholding the posterior probabilities of edges involving node $i^*$, that is,  
\begin{equation}\label{UpdateAtest}
A^\mathrm{swap}_{i^*,j} = A^\mathrm{swap}_{ j,i^*} =   \mathds{1}\{\rho^\mathrm{swap}_{i^{*},j} > 0.5\},\quad \forall j\in\dblbr p,
\qquad
A^\mathrm{swap}_{i,j}=A_{i,j},\quad  \forall i\neq i^*,j\neq i^*.
\end{equation}	
where
$		\rho^\mathrm{swap}_{i,j}  = \P_{\hat \theta}(A_{i,j}=1\mid X, Z^\mathrm{swap})$ and current parameter estimate $\hat\theta$ composed of
	\begin{align} 		\label{Updatew}  
		\wh{w}_{q,l} &= \frac{\sum_{i<j} \cZ_{q,l}^{i,j}  \rho_{i,j}}{\sum_{i<j} \cZ_{q,l}^{i,j}}, \quad
		\wh{\mu}_{q,l} &= \frac{\sum_{i<j} \cZ_{q,l}^{i,j} \rho_{i,j} X_{i,j}}{\sum_{i<j} \cZ_{q,l}^{i,j}\rho_{i,j}}, \quad 
		\wh{\sigma^2_{q,l}} &= \frac{\sum_{i<j} \cZ_{q,l}^{i,j} \rho_{i,j} (X_{i,j}-\wh{\mu}_{q,l})^2}{\sum_{i<j} \cZ_{q,l}^{i,j}\rho_{i,j}},
	\end{align}	
where the posterior probabilities $\rho_{i,j}  = \P_{\hat \theta}(A_{i,j}=1\mid X, Z)$ are used as weights to define estimates that have the form of weighted means.  

In the greedy algorithm, as the variation $\Delta_{g\rightarrow h}$ is evaluated many times, its computation   represents a huge part of the computational burden of the algorithm. Thus,  it is important that  the numerical evaluation of $\Delta_{g\rightarrow h}$ is fast. Details on the efficient computation of $\Delta_{g\rightarrow h}$ are provided in the Supplementary Material.

It occurs that the algorithm successively moves all nodes from one block to other blocks resulting in an empty block.  Thus, starting from a large number $Q_\text{init}$ of blocks, the approach automatically performs model selection by emptying blocks and providing a final clustering with an appropriate number of clusters. The  ICL criterion can be considered as a penalized criterion encouraging the use of a parsimonious model.

The greedy ICL algorithm is summarized in    Algorithm \ref{GreedyICL}. The estimates of  the model parameter $\theta$ and the clustering $Z$ are then used to compute the $\ell$-values (obtained in fact directly from the posterior edges probabilities $\rho$) to infer a graph such that the FDR is controlled. The code will soon be available in the R package  \texttt{noisysbmGGM}. 

A numerical comparison of the method with the variational EM algorithm implemented in the package \texttt{noisySBM} showed that both methods provide essentially the same estimates when applied to the same dataset, while the greedy ICL algorithm largely outperforms   EM   in terms of computing time. Furthermore, graphs of larger size can be treated with this new algorithm, which makes it suitable for GGM inference.

\begin{algorithm}[h!]
	\SetAlgoLined
	\SetKwInOut{Input}{Input}
	\SetKwInOut{Output}{Output}
	\Input{Observation  $X$,   initial number of latent blocks $Q_\mathrm{init}$}
	\Output{Clustering $Z$ into $Q$ clusters, estimator $\theta$, edge probabilities $\rho$}
	Initialize $Z$, $\theta$, $\rho = (\rho_{ij})_{i,j \in \cA}$ and $A$; 	
	Set $Q=Q_\mathrm{init}$\;
	\Repeat {{convergence}}{
	 Select a node $ i^{*} \in \dblbr  p$ ; 
	 Set $g \leftarrow Z_{i^{*}}$\;
		\For{$h \ne g$}{
			Calculate $\Delta_{g \rightarrow h}$\; 
		}
		Find $h$ such that  $\Delta_{g \rightarrow h}$ is maximum\;
		\If{ $\Delta_{g \rightarrow h}>0$}{			
			Update clustering $Z$ by setting  $Z_{i^{*}}=h$\;
			\If{block $g$ is empty}{
				Shrink the number of blocks: $Q=Q-1$\;
				}
			Update $\theta$, $\rho = (\rho_{ij})_{i,j \in \cA}$ and $A$ based on the updated clustering $Z$; 
		}
	}
	\caption{Greedy ICL algorithm for the NSBM} \label{GreedyICL}
\end{algorithm}

 \section{GGM inference using the NSBM approach} \label{sec:GGM}
This section describes how to use    the multiple testing approach using the NSBM to infer a GGM.
In particular, different test statistics from the GGM literature are presented, which are convenient for  building the observation matrix $X$ provided to the NSBM algorithm. 

\subsection{GGM test statistics for the NSBM}
Let $\mathbf Y = (Y_{1},  \ldots, Y_{p})^{t}$ be a  centered random vector with   multivariate  Gaussian distribution with   nonsingular covariance matrix $\Sigma$.  
The inverse of the latter,  denoted by $\Omega=\Sigma^{-1}=(\omega_{i,j})_{i,j}$, 
 is called the precision matrix and represents the partial correlations. More precisely,  
 $$\mbox{Corr}(Y_i, Y_j|\mathbf Y_{-(i,j)}) = - \frac{\omega_{i,j}}{\sqrt{\omega_{i,i} \omega_{j,j}}},\quad\text{where } \quad \mathbf Y_{-(i,j)}= \{Y_k, k \notin \{i,j\}\}.$$
The interpretation is that $\omega_{i,j}=0$ if and only if   $Y_{i}$  and $Y_{j}$ are   independent  given    all other variables $\mathbf Y_{-(i,j)}$. In other words, if $\omega_{i,j}\neq0$, then there is a direct   dependency or link between $Y_i$ and $Y_j$ \citep{Lauritzen}.
Now, the aim of GGM inference  is  to detect 
the pairs of random variables  with nonzero partial correlation  on the basis of a sample of  $\mathbf Y$. This amounts to infer the binary graph with adjacency matrix $A$ with $A_{i,j}=\1\{\omega_{i,j}\neq0\}$.

The literature provides various test statistics computed on a sample of size $n$ of $\mathbf Y$ 
for testing whether partial correlations are zero or not. They can be used to directly build a multiple testing procedure, but in this work we use them to build the observation matrix $X$ fed to the greedy NSBM algorithm. The following test statistics are considered:
 \begin{itemize}
\item When the sample size $n$ exceeds the number of variables $p$,   the inverse of the sample covariance matrix is an estimate of the precision matrix. Applying  Fisher's $z$-transformations  to the associated partial correlations  provides
asymptotically normal  test statistics \citep{Anderson}.
\item  \cite{Ren2015}
use  the residuals of  bivariate regressions that model the outcome of two variables $Y_i$ and $Y_j$ as a function of all   other variables to estimate the precision matrix. We refer to this test statistic as \texttt{Ren}. Asymptotic normality is obtained in the sparse setting, where the maximum degree $s$ of the graph satisfies $s = o(\sqrt{n}/\log p)$.
\item \cite{Jankova2018} propose test statistics inspired by the debiasing approach in high-dimensional linear regression. Two variants are available, the first one uses graphical Lasso, here referred to as \texttt{JankovaGL}, the second is a nodewise Lasso approach, denoted by  \texttt{JankovaNW}. In the same sparse regime as above, these test statistics are asymptotically normal. 
\item \cite{Liu2013}  uses a bias correction of the sample covariance of residuals when considering the conditional regression for the variables $Y_i$ and $Y_j$.
Two variants are available depending on the method for the estimation of the regression parameters:  Lasso regression and a scaled Lasso regression, denoted by \texttt{LiuSL} resp. \texttt{LiuL}. Contrary to the other methods,  the asymptotic variance of these  test statistics is unknown. Thus, we will apply an extension of Algorithm~\ref{GreedyICL} that estimates the limit variance (see Supplementary Material). 
\end{itemize}

\subsection{Multiple testing procedures for GGM inference}\label{subsec:proc}
\cite{Liu2013}     proposes   a  classical multiple testing approach   for GGM inference based   on the   test statistics he proposed and an adjustment of the significance level of the tests. The approach achieves  the asymptotic control of  the FDR at the desired level $\alpha$ in the sparse setting. It is straightforward to use Liu's procedure with all other test statistics described above and this is done in our numerical study in the next section, where Liu's multiple testing approach  serves as a baseline to which our NSBM approach will be  compared. 
The main difference of Liu's approach and ours is that Liu only applies a correction of the significance level, but 
 does not model neither the joint distribution of the test statistics nor the graph structure. 

In detail, in the numerical study we consider the challenging setting where $p>n$.
We denote by \texttt{Ren-classical}, \texttt{JankovaNW-classical}, \texttt{JankovaGL-classical}, \texttt{LiuSL-classical}, \texttt{LiuL-classical} the classical multiple testing procedure by \cite{Liu2013} applied to the different  test statistics. Computations are performed with the R package \texttt{SILGGM} \citep{SILGGM}. Then, we apply the analogous  procedures with our NSBM approach, that is Algorithm  \ref{GreedyICL}  applied to the same test statistics, followed by the multiple testing procedure based on $\ell$-values. We denote them by 
 \texttt{Ren-NSBM}, \texttt{JankovaNW-NSBM}, \texttt{JankovaGL-NSBM}, \texttt{LiuSL-NSBM}, \texttt{LiuL-NSBM}. For  \texttt{LiuSL-NSBM} and \texttt{LiuL-NSBM} the version of the NSBM with unknown variances under the alternative (see Section 2 of the Supplementary Material)  is used. This is done with our R package \texttt{noisysbmGGM}. 
Furthermore,   the following popular procedures that do not control any error rate are applied:
 the procedure by \cite{MB} (\texttt{MB}) and graphical Lasso (\texttt{glasso}). Finally, the CLIME procedure (\texttt{clime}) \citep{Clime, Transclime} is applied to complete the picture.

\section{Numerical study}\label{sec:Simus}

In this section  various numerical experiments are conducted  to assess the performance of our NSBM approach for GGM inference, to compare to the state of the art, and to evaluate robustness to model misspecification.

\subsection{Comparison of   test statistics  on synthetic data}

\begin{figure}[t]
\subfloat[SBM]{
\includegraphics[width=.22\textwidth]{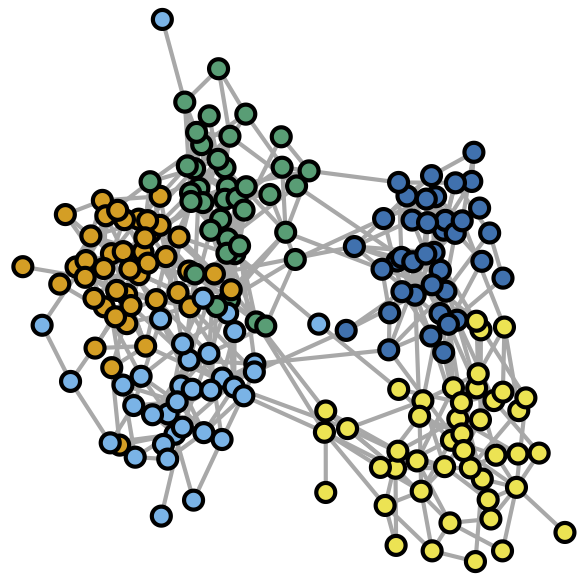}}\quad
 \subfloat[Hub]{
\includegraphics[width=.22\textwidth]{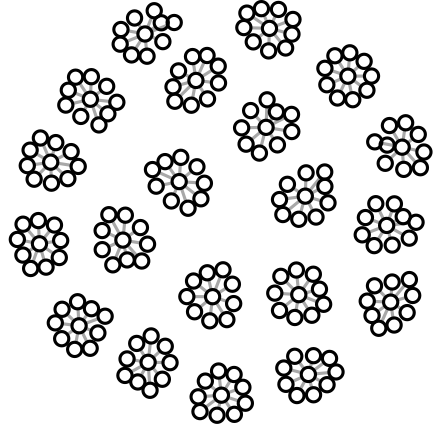}}\quad
 \subfloat[Band]{
\includegraphics[width=.22\textwidth]{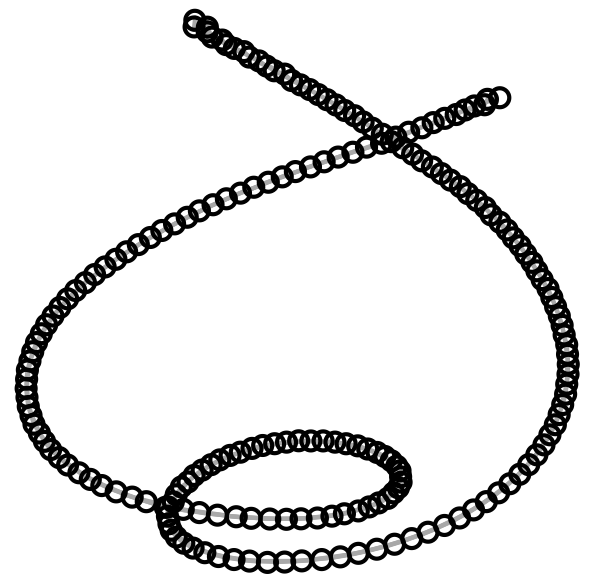}}\quad
\subfloat[Scale-free]{
\includegraphics[width=.22\textwidth]{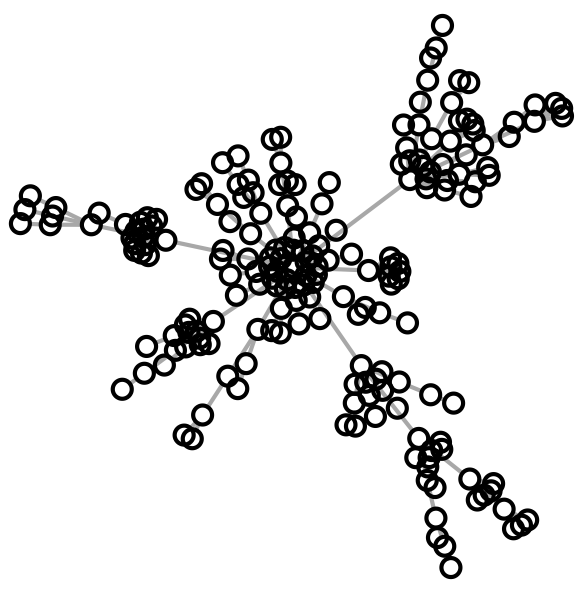}} 
\centering
\caption{Illustration of the  graph structures in the four settings.}
\label{Allgraphs}
\end{figure}

\begin{figure}[h!]
	\includegraphics[scale=1]{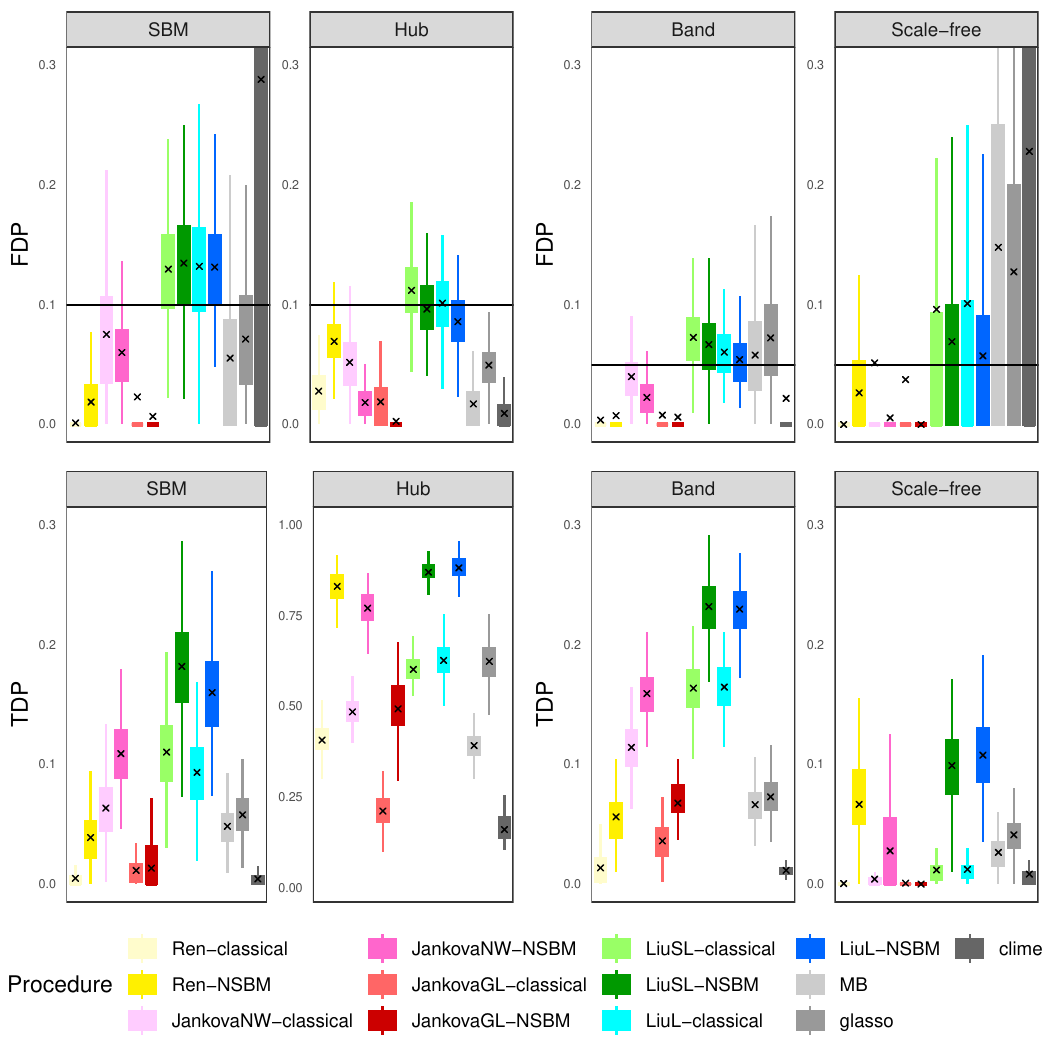}	
	\centering
	\caption{Boxplots of the FDP and  TDP for all procedures described in Section \ref{subsec:proc}
	for the settings SBM, hub, band and scale-free for 200 simulated datasets, where $p=200$ and $n=100$. Horizontal lines  represent the nominal level $\alpha$, with $\alpha=0.1$   in  the SBM and the hub setting 
and $\alpha=0.05$ in the band and scale-free cases.
  The crosses in the boxplots correspond to the sample FDR  (resp. TDR) defined as the mean of the FDP (resp.  TDP).
} 
	\label{fig:simp200}
\end{figure}

Data are generated from a Gaussian distribution  $\mathcal N_p(0, \Sigma)$ with different choices for the precision matrix $\Omega=(\omega_{i,j})_{i,j}=\Sigma^{-1}$. More precisely, we consider different   underlying graphs $A=(\1\{\omega_{i,j}\})_{i,j}$ illustrated in  Figure \ref{Allgraphs}. 
In  the SBM setting (a), the graph $A$ is the realization of a SBM with five communities of equal size and varying probabilities of connection among the blocks. To evaluate the robustness of our procedure  to model misspecification, we also consider   three settings where $A$ is not a SBM. In the hub setting~(b), the set of  nodes is partitioned into   groups of size 10. In each group one node is  a hub and as such connected to all other nodes in the group. In the band setting (c), every node $i$ is connected to all nodes $j$ at distance 3, that is $A_{i,j}=\1\{1\leq|i-j|\leq 3\}$. 
This yields a matrix $A$ with  bands of ones around the diagonal. Finally, in the scale-free setting (d), the graph is generated by preferential attachement   \citep{Barabasi}. That is, the graph is built iteratively starting from a small   initial graph, here  consisting of only two connected nodes. Then new nodes are added one by one and  each of them is connected to a single  node  in the graph with  probabilities proportional to the current  node degrees.



Now, to generate a precision matrix $\Omega$ and a covariance matrix $\Sigma$ associated with a given adjacency matrix $A$, we proceed as in the R package {\verb|huge|}. That is, first a positive definite  matrix $\Omega$ is obtained by the  transformation
\begin{equation*}
		\Omega=\gamma A + \text{diag}\left(|\min_{\lambda\in \text{Sp}(\gamma A)}(\lambda)|+\beta\right),
\end{equation*} 
	where  $\gamma =0.3$, $\beta=0.2$ and $\text{Sp}(\gamma A)$ denotes the spectrum of $\gamma A$.
In other words,  the  ones in $A$ are replaced by  $\gamma$ and the diagonal is such that  $\Omega$ is invertible. Hence,  $\Sigma=\Omega^{-1}$ is well defined and realizations  of the desired Gaussian distribution can be simulated. 

In the simulations there are  $p=200$ variables and the sample size is $n=100$.  In every setting all procedures described in Section \ref{subsec:proc} are evaluated on 200 datasets with  nominal level $\alpha=0.1$ in the SBM and the hub setting 
and $\alpha=0.05$ in the band and scale-free cases.
The performance of a  procedure is evaluated by comparing  the inferred graph $\hat{A}$ to  the true adjacency matrix $A$. 
In particular, the  false discovery proportion (FDP) is the proportion of errors among the edges declared as significant, whereas   the proportion of edges that are correctly declared as significant among the true edges in $A$  is referred to as the true    discovery proportion (TDP). They are defined as 
$$
\mathrm{FDP}=\frac{|\{(i,j): \hat A_{i,j}=1, A_{i,j}=0\}|}{|\{(i,j): \hat A_{i,j}=1\}| \vee 1},\qquad
\mathrm{TDP}=\frac{|\{(i,j): \hat A_{i,j}=1, A_{i,j}=1\}|}{|\{(i,j):   A_{i,j}=1\}|}.
$$
For procedures that are supposed to control the FDR, the FDP   should in average  be lower or equal than $\alpha$, while the TDP representing the power of the  procedure is desired to be as large as possible.

  Figure \ref{fig:simp200} displays the results. 
First,   we observe that concerning the FDR,  for all Ren- and Jankova-procedures the FDR is always controlled at the nominal level $\alpha$. The FDR of all Liu-procedures slightly exceeds $\alpha$, but the Liu-NSBM procedures mostly do better than their  classical inference counterparts. The last three procedures in the figure  
are not conceived for controlling the FDR, so the comparison to the nominal level does not really make sense since the output of these procedures is the same for any $\alpha$.

Second, in terms of power, in every setting the largest TDRs are achieved by NSBM based procedures, in particular our approach with both of Liu's test statistics appears to be the most powerful procedures. 
Notably, also the  power of the  methods that do not require the control of the FDR, namely the one by Meinshausen \& Bühlmann,   graphical Lasso and CLIME, is   largely below those of most other procedures. 
Moreover, in all settings the NSBM approach is more powerful than the classical inference procedure using the same test statistics. 
That is, whatever test statistic is used, there is a gain in power when using our NSBM approach. 
This is in line with what has been observed for other  multiple testing procedures, that is that
using a latent structure and 
   learning the model   is beneficial  \citep{SC09, cs09, Liu2016, RRV22}.
Furthermore, these results    confirm the observation done in \cite{RRV22}  that the SBM is an appropriate choice for graph inference, as the SBM is a highly flexible random graph model accommodating  a wide spectrum of graph topologies. 
Here we see that the SBM is   also appropriate  in the framework of   GGM inference. 
In other words, our approach is robust to model misspecification as good results are also obtained in the hub, band and scale-free settings in our study.  

  \subsection{Sparse setting}

As in practice, most graphs are rather sparse, we now investigate the behaviour of our procedure in the sparse setting described in terms of   the maximum degree in the graph. Indeed, most of the test statistics that we consider are known to be asymptotically normal in the sparse setting when $s = o(\sqrt{n}/\log(p))$,  where $s$ is the maximum node degree in the graph. 
For this, 
we simulate data  with given maximum degree.
More precisely,  given a maximum degree $s$, we first generate a sequence of degrees according to a power law with exponent $2$ and then try to simulate a graph $A$ with this degree sequence. 
Then, we proceed as previously to generate Gaussian vectors. 

  \begin{figure}[t]
	\includegraphics[scale=.7]{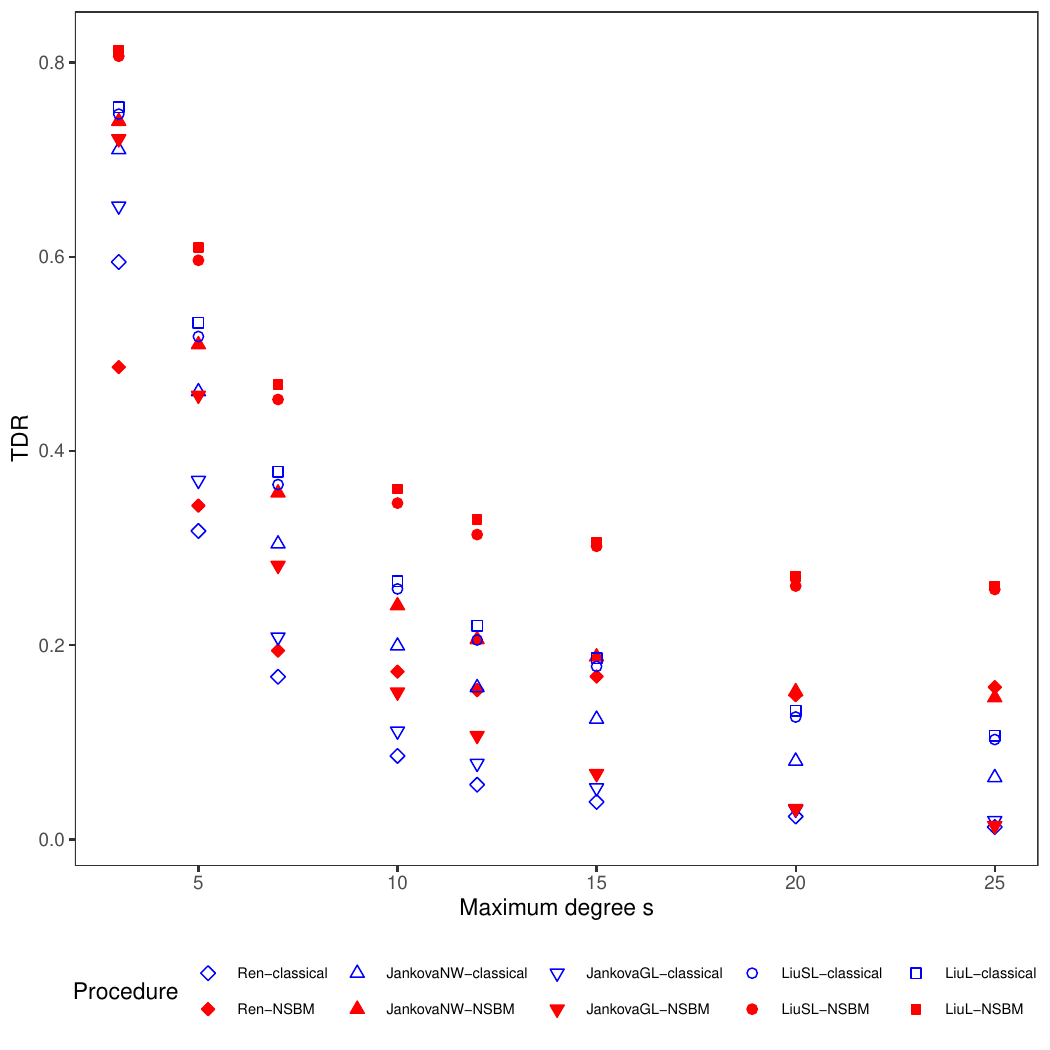}
	\centering
	\caption{TDR (Average TDP)  evaluated  for different procedures over 200 simulations versus the maximum degree $s$ of the graphs ($p=100, n=100, \alpha=0.1$). In  blue and filled red are respectively  represented the classical inference procedures and  our NSBM procedures with the different test statistics.}
		\label{maxdegree}
\end{figure}

We carry out simulations with graphs composed of $p = 100$ nodes and $n=100$ observations, for a maximum degree $s$ in the range from 3 to 20.  The average FDR and average TDR are estimated over 200 simulations and results are presented in Figure \ref{maxdegree}. 
As expected,   the larger the maximum degree $s$, the lower the power. This is true for all  methods, but the  decrease in power of the NSBM  procedures is less important than for the classical ones.
Thus, there is a clear  benefit of using our procedure for GGM inference in sparse settings.

\subsection{Validation on   a human T-cell dataset}

Now, we investigate the performance of our  procedure on a real dataset. 
We consider   the multivariate flow cytometry data produced by   \cite{Sachs}  concerning a human T-cell signalling pathway, whose deregulation may lead to carcenogenesis.
 The data are available from \url{https://www.science.org/doi/10.1126/science.1105809} and were extensively studied in the literature  \citep[see e.g.][]{Sachs, Heneo11, Desgranges15,Eaton07, Sachs18}. 
The amount of 11 specific proteins  in single cells was simultaneously measured  under different perturbation conditions of the cell. 
We focus on a general perturbation (anti-CD3/CD28 + ICAM-2) that overall stimulates the cellular signalling network. In this condition measurements for $902$ cells are available. 

As the dataset is rather huge,  it can be used to first compute a benchmark on the full dataset and then assess the performance of the inference procedures on small subsets of the data. More precisely, on the full dataset   various methods, namely the classical procedure with Liu's test statistics \texttt{LiuL} and nominal  level $\alpha=0.05$, yield
  the graph displayed in Figure \ref{inferredSachs}. The inferred  graph contains ten edges,  nine  of which are   interactions well known from the literature. The additional edge p38-JNK  was also  inferred by \cite{Sachs} but with low confidence, and  also was found by  \cite{Heneo11, Desgranges15, Eaton07}.

\begin{figure}[t]
	\begin{center}
		\includegraphics[width=.6\textwidth]{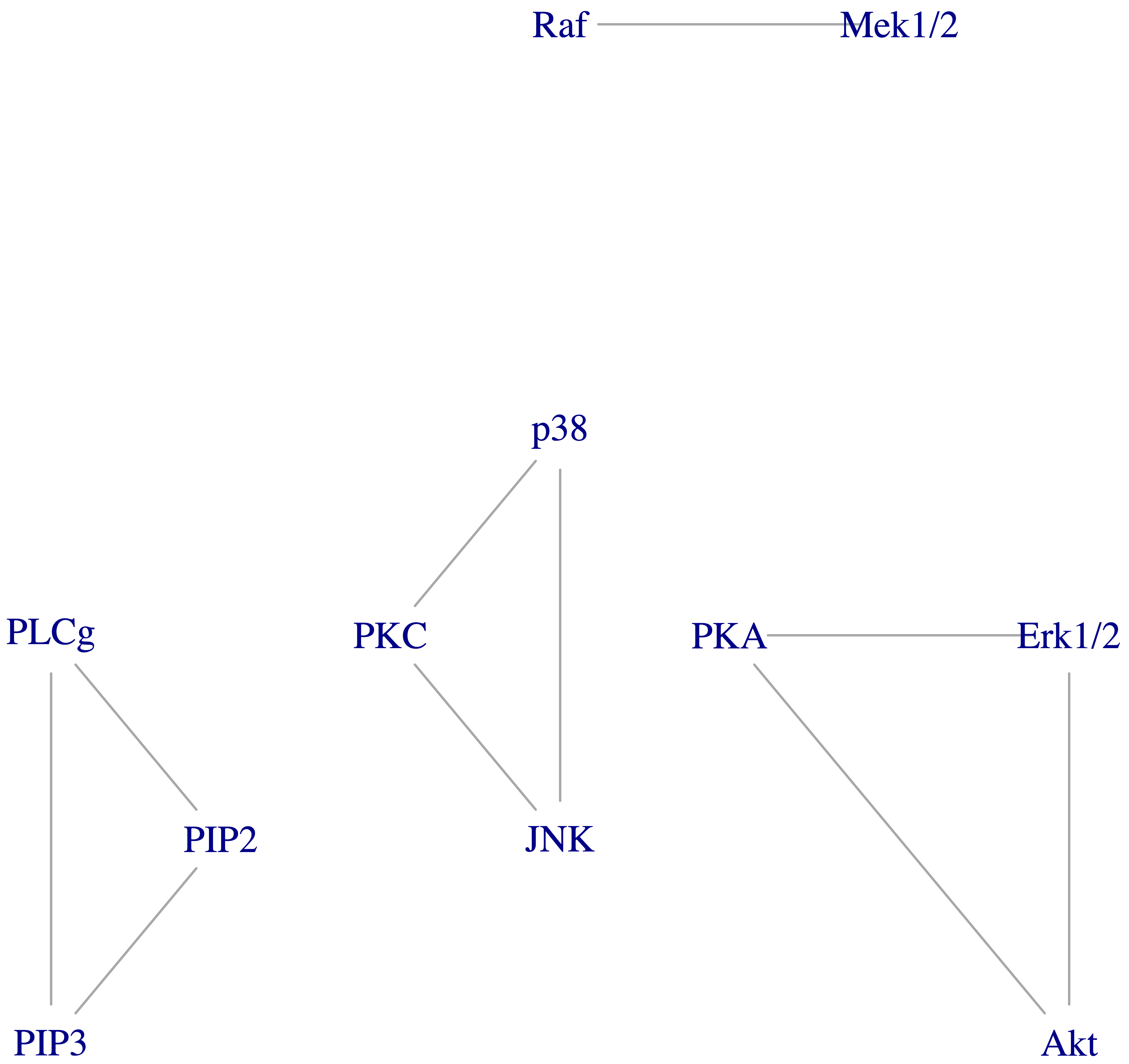}
	\end{center}
	\caption{Benchmark obtained on  the full dataset.}
	\label{inferredSachs}
\end{figure}


\begin{table}[t]
	\center
	\begin{tabular}{l|c|c|}
		\multicolumn{1}{c|}{edge} & \texttt{LiuL-classical}   & \texttt{LiuL-NSBM}     \\
		\hline 
		Raf -  Mek1/2 &  200 & 200 \\
		PLCg - PIP2   & 28  &44 \\ 
		PLCg - PIP3 & 121  &143\\ 
		PIP2 - PIP3   &  167 &172 \\
		Erk1/2 - Akt &199  &199 \\
		Erk1/2 -PKA  &42 &54 \\
		Akt - PKA & 92 &121 \\ 
		PKC - p38 & 146  &158 \\
		PKC - JNK & 107  &121 \\
		p38 - JNK &102 &115 
	\end{tabular} 
	\caption{The number of times the 10 edges of Figure \ref{inferredSachs} are detected over 200 subsampled datasets with sample size $n=20$ for both the classical and the NSBM procedure with Liu's test statistics \texttt{LiuL} with $\alpha=0.05$.}
	\label{TableSachs}
\end{table}

Now to  evaluate the ability   to recover these edges from a smaller dataset, we sample  subsets of 
$n=20$ cells 
and apply both the classical  and the NSBM procedures with the test statistic \texttt{LiuL} with $\alpha=0.05$.
Table~\ref{TableSachs}  reports   the number of times the edges of the benchmark network from Figure \ref{inferredSachs} are detected with each method on 200 sampled datasets.
Obviously, the  NSBM procedure detects all edges more often than the classical procedure. That is, there is a real gain in power also on real datasets. 
This confirms that using our procedure helps to improve GGM inference.

\section{Conclusion}
In this work a new approach for GGM inference is proposed that relies on the modeling of the graph structure by a latent variable model. It allows to control the false discovery rate and outperforms    state-of-the-art methods in terms of power, that is, more ``true'' edges are detected than with classical methods. 
The performance is assessed in various numerical experiments and the procedure is particularly suited for real datasets, since  it is  is robust to model misspecification and works on sparse networks as well as on small datasets.
 
It is noteworthy that this is a generic approach 
that can be used in combination with any test statistic for GGM inference. 
It can be viewed as a layer on top of the traditional problem of constructing a test statistic. 
So the user may consider his or her favorite test statistic and provide them to our NSBM method  for GGM inference. We have shown that this consistently yields more powerful results than with the classical inference approach. 

Moreover, the new greedy ICL algorithm for fitting the NSBM to a dataset is much more attractive in terms of computing time and  scalability to large graphs than the variational EM algorithm. It is noteworthy that the ICL algorithm can be used beyond the context of NSBM-based GGM inference.

\bibliographystyle{Chicago}

\bibliography{biblio}
\end{document}


\def\spacingset#1{\renewcommand{\baselinestretch}%
{#1}\small\normalsize} \spacingset{1}



\if0\blind
{
  \title{\bf Supplement to "Enhancing the Power of Gaussian Graphical Model Inference by Modeling the Graph Structure"}
  \author{Valentin Kilian, Tabea Rebafka and Fanny Villers\thanks{
    The authors gratefully acknowledge \textit{please list all relevant funding sources in the unblinded version}} 
        }
        \date{LPSM, Sorbonne Université, Paris, France}
  \maketitle
} \fi

\if1\blind
{
  \bigskip
  \bigskip
  \bigskip
  \begin{center}
    {\LARGE\bf Title}
\end{center}
  \medskip
} \fi

\bigskip
\begin{abstract}
This document provides additional information and details on the article "Enhancing the Power of GGM Inference by Modeling the Graph Structure". First, all mathematical details on the new greedy ICL algorithm for the NSBM are provided. Then another version of the NSBM where variances of the observations are supposed to be unknown is introduced and the associated inference algorithm is detailed. Moreover, further numerical results are presented to support the conclusions on the   performance of the NSBM approach for GGM inference.
 \end{abstract}

\spacingset{1.75} 

\section{Greedy algorithm}
  \label{appendix:Proof}

\subsection{Proof of Proposition 1}
The expression of $\mathrm{ICL^{SBM}}$ in   Proposition 1 of the main paper
is proven in \cite{Come15}. 
So we only prove the expression of $\mathrm{ICL^{noise}}(X,A,Z)$, 
which corresponds to the term of the ICL related to the noise layer of the NSBM. With the notations of the main paper,
\begin{align*}
	p(X|A, Z, {\mu}) &=  \prod_{(i,j)\in\mathcal{A}} p(X_{i,j} | A_{i,j}, Z_i, Z_j, {\mu})\nonumber\\
	&= \prod_{(i,j)\in\mathcal{A}} \left\{ \left[f_{\mathcal{N}(0,1)}(X_{i,j})\right]^{1-A_{i,j}} \prod_{  q\leq l} \left[f_{\mathcal{N}(\mu_{q,l},\sigma^2)}(X_{i,j})\right]^{A_{i,j}\mathcal{Z}_{q,l}^{i,j}}\right\}.  
\end{align*}
By integration,
\begin{align*}
&	\mathrm{ICL^{noise}}(X,A,Z) = \log \left( \int p(X|A, Z, { \mu})\boldsymbol \pi_\mu ({ \mu})d{ \mu}\right)\\
	&= \log \left(\prod_{(i,j)\in\mathcal{A}}  \left[f_{\mathcal{N}(0,1)}(X_{i,j})\right]^{1-A_{i,j}} \times
	\prod_{ q\leq l}  \int  \left\{\prod_{(i,j)\in\mathcal{A}}\left[f_{\mathcal{N}(\mu_{q,l},\sigma^2)}(X_{i,j})\right]^{A_{i,j}\mathcal{Z}_{q,l}^{i,j}}\right\} 
 f_{\mathcal{N}(\rho_0, \tau_0^2)}(\mu_{q,l})d\mu_{q,l}  \right)\\	
	&= \sum_{(i,j)\in\mathcal{A}} (1-A_{i,j}) \log (f_{\mathcal{N}(0,1)}(X_{i,j})) + 
	\sum_{q\leq l} \log (\mathcal{J}_{q,l}),
\end{align*}
where
\begin{align*}
 \sum_{(i,j)\in\mathcal{A}} (1-A_{i,j}) \log (f_{\mathcal{N}(0,1)}(X_{i,j}))  
 &= -\frac12 \sum_{(i,j)\in\mathcal{A}} (1-A_{i,j}) X_{i,j}^2 - \frac {N - M^A}2 \log(2\pi) ,
\end{align*}
and, for   $q\leq l$,
\begin{align*}
\mathcal{J}_{q,l} 
&=\int  \prod_{(i,j) \in I_{q,l}} f_{\mathcal{N}(\mu_{q,l},\sigma^2)}(X_{i,j})\times f_{\mathcal{N}(\rho_0, \tau_0^2)}(\mu_{q,l})d\mu_{q,l}\\
	&= \int  \left(\frac{1}{\sigma\sqrt{2\pi}}\right)^{n_{q,l}}\exp\left\{-\sum_{(i,j)\in I_{q,l}}\frac{(X_{i,j}-\mu_{q,l})^2}{2\sigma^2}\right\} \times \frac{1}{\tau_0\sqrt{2\pi}}\exp\left\{-\frac{(\mu_{q,l}-\rho_0)^2}{2\tau_0^2}\right\}d\mu_{q,l}\\
	&= \left(\frac{1}{\sigma\sqrt{2\pi}}\right)^{n_{q,l}}\frac{1}{\tau_0\sqrt{2\pi}}\int \exp\left\{-\frac{1}{2}\left[\frac{1}{\sigma^2}\sum_{(i,j)\in I_{q,l}}(X_{i,j}-\mu_{q,l})^2 + \frac{(\mu_{q,l}-\rho_0)^2}{\tau_0^2}\right]\right\}d\mu_{q,l}.
\end{align*}
Developing the term in brackets, it turns out that it is equal to 
\begin{align*}
	\mu_{q,l}^2\underbrace{\left(\frac{1}{\tau_0^2}+\frac{n_{q,l}}{\sigma^2}\right)}_{=a}&-2\mu_{q,l}\underbrace{\left(\frac{\rho_0}{\tau_0^2}+\frac{1}{\sigma^2}\sum_{(i,j)\in I_{q,l}}X_{i,j}\right)}_{= b}+\underbrace{\left(\frac{\rho_0^2}{\tau_0^2}+\frac{1}{\sigma^2}\sum_{(i,j)\in I_{q,l}}X_{i,j}^2\right)}_{=c}\\
	&=a\left(\mu-\frac{b}{a}\right)^2 -\frac{1}{a}(b^2-ac).
\end{align*}
Consequently, by integrating we obtain for  $\mathcal{J}_{q,l}$,
\begin{align*}
	\mathcal{J}_{q,l} 
	&= \left(\frac{1}{\sigma\sqrt{2\pi}}\right)^{n_{q,l}}\frac{1}{\tau_0\sqrt{a}}\exp\left\{\frac{1}{2a}( b^2-ac)\right\}.
\end{align*}
One can show that 
\begin{align*}
	b^2-ac 
	& = -\frac{1}{\tau_0^2\sigma^2}\sum_{(i,j)\in I_{q,l}}(X_{i,j}-\rho_0)^2 -\frac{n_{q,l}^2}{\sigma^4}S_{q,l}.
\end{align*}
Hence,
\begin{align*}
	\mathcal{J}_{q,l}  
		=&\left(\frac{1}{\sigma\sqrt{2\pi}}\right)^{n_{q,l}}\left(\frac{\sigma^2+\tau_0^2n_{q,l}}{\sigma^2}\right)^{-\frac12} \\
		 &\times\exp\left\{- \frac{1}{2(\sigma^2+ \tau_0^2n_{q,l})} \sum_{(i,j)\in I_{q,l}}(X_{i,j}-\rho_0)^2 -\frac{\tau_0^2n_{q,l}^2}{2(\sigma^4 + \tau_0^2\sigma^2n_{q,l})}S_{q,l} \right\}.
	\end{align*}
Then,
\begin{align*}
\sum_{q\leq l}\log	\mathcal{J}_{q,l}
&= -\log\sigma \left(\sum_{q\leq l}n_{q,l}-N_Q\right)-\frac12\log(2\pi)\sum_{q\leq l}n_{q,l}
-\frac12\sum_{q\leq l}\log\left( \sigma^2+{\tau_0^2n_{q,l}}\right) \\
&- \frac12\sum_{q\leq l}\frac{1}{ \sigma^2+ \tau_0^2n_{q,l}} \sum_{(i,j)\in I_{q,l}}(X_{i,j}-\rho_0)^2 
-\frac12\sum_{q\leq l}\frac{\tau_0^2n_{q,l}^2}{ \sigma^4 + \tau_0^2\sigma^2n_{q,l}}S_{q,l}.
\end{align*}
By gathering all terms and noting that
$M^A= \sum_{q\leq l}n_{q,l}$, we obtain the desired result.

 \subsection{Efficient computation of $\Delta_{g\rightarrow h}$}
 
 In the greedy algorithm  the impact on the ICL of swapping a node $i^{*}$ from its current block, say $g$, to another block $h$ must be evaluated. That is, one has to compute
 $$\Delta_{g\rightarrow h} = \mathrm{ICL^{NSBM}}(X,Z^\mathrm{swap}, A^\mathrm{swap}) - \mathrm{ICL^{NSBM}}(X,Z, A).$$
 The numerical evaluation of  $\mathrm{ICL^{NSBM}}$ has some computational cost, however the difference $\Delta_{g\rightarrow h}$ can be computed more efficiently, as changing the block assignment of a single node has an impact only   on a small number of the count statistics that appear in the expression of the  $\mathrm{ICL^{NSBM}}$  in Proposition 1 of the main paper. 
 For instance, the number of  nodes associated with block $q$, that is
 $n_q=  \sum_{i\in\dblbr n} Z_{i,q}$, only changes for $q=g$ and $q=h$.
 Let us add superscript $^\text{swap}$ to denote the quantities after the swap. As such we have $n_g^\text{swap}=n_g-1$, $n_h^\text{swap}=n_h+1$ and $n_q^\text{swap}=n_q$ for all $q\notin\{g,h\}$.
 
 For the count statistics $n_{q,l}$, 
it is clear that their value changes only if $q\in\{g,h\}$ or $l\in\{g,h\}$. And as $A$ and $A^\text{swap}$ are identical except on the $i^*$-th row and $i^*$-th column,   changes only come from edges involving node $i^*$. 
We introduce the sets $I_l^A=\{j :  Z_j=l, A_{i^*,j}=1\}$ indicating the set of nodes assigned to block $l$ to which $i^*$ is connected in graph $A$. 
 Then one can show that $n^\text{swap}_{g,l}=n^\text{swap}_{l,g}=n_{g,l}-  |I_{l}^A| $ for all $ l\neq h$, $n^\text{swap}_{h,l}=n^\text{swap}_{l,h}=n_{h,l}+  |I_{l}^{A^\text{swap}}| $ for all $ l\neq g$, and $n^\text{swap}_{g,h}=n^\text{swap}_{h,g}=n_{g,h}-  |I_{h}^A|+  |I_{g}^{A^\text{swap}}|$. If both $q\notin\{g,h\}$ and $l\notin\{g,h\}$, then $n^\text{swap}_{q,l}=n_{q,l}$.

  Concerning the changes in    $\bar n_{q,l}$, first note that   $ n_{q,l}+\bar n_{q,l}$ is the maximal number of possible edges between nodes in $q$ and $l$, that we may denote by $m_{q,l}$ and which can also be  written as $m_{q,l}=n_qn_l$ if $q\neq l$ and $m_{q,q}=n_q(n_q-1)/2$ for $q\in\dblbr Q$. We have , for all $l\notin\{g,h\}$,
  $m_{g,l}^\text{swap}=m_{l,g}^\text{swap}=m_{g,l}-n_l$ and   $m_{h,l}^\text{swap}=m_{l,h}^\text{swap}=m_{h,l}+n_l$. Moreover, $m_{g,g}^\text{swap}=m_{g,g}-n_g+1$, $m_{h,h}^\text{swap}=m_{h,h}+n_h$ and $m_{g,h}^\text{swap}=m_{h,g}^\text{swap}=m_{g,h}-n_h+n_g-1$.  
Then, $\bar n_{q,l}^\text{swap} = m_{q,l}^\text{swap}-n_{q,l}^\text{swap}$.
     
Likewise, only the sets $I_{q,l}$ with $q\in\{g,h\}$ or $l\in\{g,h\}$ are changed by the swap. More precisely,
 $$I_{q,l}^\text{swap} =\left\{I_{q,l}\backslash \{\1_{\{q=g\}}\odot I_l^A\cup \1_{\{l=g\}}\odot I_q^A\}
\right\}\bigcup\left\{  \1_{\{q=h\}}\odot I_l^{A^\text{swap}}\cup \1_{\{l=h\}}\odot I_q^{A^\text{swap}} 
\right\},
$$
where $1\odot I=I$ and $0\odot I=\emptyset$ for any set $I$.

With these observations on the count statistics, it is clear that the computation of $\Delta_{g\rightarrow h}$ can be simplified, since 
the difference of the  sums taken over all $(q,l)$ such that $q\leq l$ in  $\mathrm{ICL^{NSBM}}$ reduces to   sums involving only quantities that change with the swap of $i^*$ from $g$ to $h$. More precisely, only the terms with pairs of indices $(q,l)$ in the   set  $\mathcal S_{g,h}$  are relevant, where $\mathcal S_{g,h}$ is defined as
$$\mathcal S_{g,h} =\left\{ (g,l):   l\in\dblbr Q \right\}\bigcup \left\{(h,l):   l\in\dblbr Q \backslash\{g\}
\right\}.
$$
     
     Now, to state  $\mathrm{ICL^{NSBM}}$  explicitly, two cases must be distinguished depending on whether the swap empties block $g$ (that is $i^*$ is the last node assigned to block $g$) or not.

 \begin{proposition}  
 		\label{DeltaICLswap}
For $A^\mathrm{swap}$ defined  in (4) 
		in the main paper,
 		\begin{equation*}
 			\Delta_{g \rightarrow h} = \Delta^\mathrm{SBM}_{g \rightarrow h} +  \Delta^\mathrm{noise}_{g \rightarrow h}, 
 		\end{equation*}
where $\Delta^\mathrm{SBM}_{g \rightarrow h} =\mathrm{ICL^{SBM}}(A^\mathrm{swap},Z^\mathrm{swap})-\mathrm{ICL^{SBM}}(A,Z)$ refers to the change of $\mathrm{ICL^{SBM}}$ induced by the swap and  
 $\Delta^\mathrm{noise}_{g \rightarrow h} = \mathrm{ICL^{noise}}(X,A^\mathrm{swap},Z^\mathrm{swap})-\mathrm{ICL^{noise}}(X,A,Z)$ the variation of the term related to the noise layer in the NSBM.
For the form of $\Delta^\mathrm{SBM}_{g \rightarrow h}$ and $\Delta^\mathrm{noise}_{g \rightarrow h}$  two cases are to be distinguished.
\paragraph{Case 1} If $i^*$ is not the only node assigned to block $g$ (that is $n_g>1$), then   moving     $i^*$   to another block does not change the number of blocks $Q$. Then 
\begin{align*}
&\Delta^\mathrm{SBM}_{g \rightarrow h} 
= \log \frac{n_0+n_h}{n_0+n_g-1} +  \sum_{ (q,l)\in\mathcal S_{g,h} }\log \frac{B(\eta_0+n_{q,l}^\mathrm{swap}, \xi_0+\bar n_{q,l}^\mathrm{swap})}{B(\eta_0+n_{q,l},\xi_0+\bar n_{q,l})},  \\
&\Delta^\mathrm{noise}_{g \rightarrow h} 
=-\frac{1}{2}\sum_{j \in\dblbr p}    X_{i^*,j} ^2 (A_{i^*,j} - A^\mathrm{swap}_{i^*,j})
- \sum_{ (q,l)\in\mathcal S_{g,h}}\left( n_{q,l}^\mathrm{swap}-n_{q,l}\right)\log\sigma\\
&\quad-\frac{1}{2} \sum_{  (q,l)\in\mathcal S_{g,h}}\log  \frac{\sigma^2+\tau_0^2n_{q,l}^\mathrm{swap}}{\sigma^2+\tau_0^2n_{q,l}}    
 -\frac{\tau_0^2}{2\sigma^2} \sum_{(q,l)\in\mathcal S_{g,h}}\left\{ \frac{(n_{q,l}^{\mathrm{swap}})^2S_ {q,l}^{\mathrm{swap}}}{\tau_0^2 n_{q,l}^{\mathrm{swap}} + \sigma^2 } - \frac{n^2_{q,l} S_{q,l} }{\tau_0^2n_{q,l}+ \sigma^2}  \right\}\\
 &\quad -\frac{1}{2}  \sum_{ (q,l)\in\mathcal S_{g,h}} \left\{ \frac{1}{\sigma^2+\tau_0^2n_{q,l}^\mathrm{swap}}\sum_{(i,j)\in I_{q,l}^{\mathrm{swap}}}(X_{i,j}-\rho_0)^2 - \frac{1}{\sigma^2+\tau_0^2n_{q,l}}\sum_{(i,j)\in I_{q,l}}(X_{i,j}-\rho_0)^2\right\}.
\end{align*}
 
 \paragraph{Case 2}
  If $i^*$ is   the only node assigned to block $g$ (that is $n_g=1$), then   
 removing $i^*$ from $g$, block $g$ disappears from the model 	 so that $Q^\mathrm{swap} = Q-1$ and
 		 		\begin{align*}
 			&\Delta^\mathrm{SBM}_{g \rightarrow h} 
			= \log\frac{n_0+n_h}{n_0}
			+ \log \left( \frac{ \Gamma(Qn_0 + p)\Gamma((Q-1)n_0)}{\Gamma((Q-1)n_0 + p) \Gamma(Qn_0)} \right)+QB(\eta_{0},\xi_{0}) \\
			  &\quad+  \sum_{l\in\dblbr Q\backslash\{g\} }\log \frac{B(\eta_0+n^\mathrm{swap}_{h,l},\xi_0+\bar n^\mathrm{swap}_{h,l})}{B(\eta_0+n_{h,l},\xi_0+\bar n_{h,l})} -  \sum_{l\in\dblbr Q}\log  B(\eta_0+n_{g,l},\xi_0+\bar n_{g,l}), \\
 			&\Delta^\mathrm{noise}_{g \rightarrow h}
			=  -\frac{1}{2}\sum_{j \in\dblbr p}    X_{i^*,j} ^2 (A_{i^*,j} - A^\mathrm{swap}_{i^*,j})
			- \left(Q+\sum_{ l\in\dblbr Q\backslash\{g\}} \left( n_{h,l}^\mathrm{swap}-n_{h,l}\right) - \sum_{ l\in\dblbr Q  } n_{g,l} \right)\log\sigma\\	
				&\quad-\frac{1}{2} \sum_{ l\in\dblbr Q\backslash\{g\} }\log  \frac{\sigma^2+\tau_0^2n_{h,l}^\mathrm{swap}}{\sigma^2+\tau_0^2n_{h,l}}
			 +\frac{1}{2} \sum_{ l\in\dblbr Q  }\log   (\sigma^2+\tau_0^2n_{g,l})\\	
			 &\quad -\frac{\tau_0^2}{2\sigma^2} \sum_{l\in\dblbr Q\backslash\{g\}}\left\{ \frac{(n_{h,l}^{\mathrm{swap}})^2S_ {h,l}^{\mathrm{swap}}}{\tau_0^2 n_{h,l}^{\mathrm{swap}} + \sigma^2 } - \frac{n^2_{h,l} S_{h,l} }{\tau_0^2n_{h,l}+ \sigma^2}  \right\}
			  +\frac{\tau_0^2}{2\sigma^2} \sum_{l\in\dblbr Q } \frac{n_{g,l}^2S_ {g,l}}{\tau_0^2 n_{g,l}  + \sigma^2 } \\
			&\quad -\frac{1}{2}  \sum_{ l\in\dblbr Q\backslash\{g\}}  \left\{\frac{1}{\sigma^2+\tau_0^2n_{h,l}^\mathrm{swap}}\sum_{(i,j)\in I_{h,l}^{\mathrm{swap}}}(X_{i,j}-\rho_0)^2 - \frac{1}{\sigma^2+\tau_0^2n_{h,l}}\sum_{(i,j)\in I_{h,l}}(X_{i,j}-\rho_0)^2\right\}\\
			 &\quad+ \frac{1}{2}  \sum_{ l\in\dblbr Q }  \frac{1}{\sigma^2+\tau_0^2n_{g,l} }\sum_{(i,j)\in I_{g,l}}(X_{i,j}-\rho_0)^2. 
 		\end{align*}
 	\end{proposition}

\begin{proof}[Proof of  Proposition \ref{DeltaICLswap}]
The  expressions of    $\Delta^\mathrm{SBM}_{g \rightarrow h}$ are given in \cite{Come15}. The formulae for $\Delta^\mathrm{noise}_{g \rightarrow h}$ are obtained by direct computation and by taking into account two particular facts.
First,   quantities related to pairs of blocks $(q,l)$ are unchanged under the swap if both indices $q$ and $l$ are different from $g$ and $h$, that is $q,l\notin\mathcal S_{g,h}$.
Second, 
 by definition, $A^\mathrm{swap}$ coincides with $A$ for all entries  not involving $i^*$, so that   corresponding terms disappear when taking the difference.
To obtain the expression in the second case, we use that 	
 $n^\mathrm{swap}_{g,l}=n^\mathrm{swap}_{l,g}=0$ for all $l\in\dblbr Q$. Then the  result is   straightforward.
\end{proof}

\subsection{Merging entire blocks}

In our numerical experiments, we observed that the solution obtained by Algorithm 1 of the main paper
can be improved, by checking whether merging  entire blocks increases the ICL  criterion. Denote by  $\Delta_{g \cup h}$  the variation of $\mathrm{ICL^{NSBM}}$ when blocks $g$ and $h$ are merged, where without loss of generality we assume that all nodes in $g$ are assigned to $h$. Then, $\Delta_{g \cup h}$  is 
  defined as 
$$\Delta_{g\cup h} = \mathrm{ICL^{NSBM}}(X, A^\mathrm{merge}, Z^\mathrm{merge}) - \mathrm{ICL^{NSBM}}(X, A,Z),$$
where $Z^\mathrm{merge}$ coincides with $Z$ except for all $i$ such that $Z_i=g$, for which we set $Z_i^\mathrm{merge}=h$, 
and $A^\mathrm{merge}$ is obtained from $A$ by modifying only the entries $(i,j)$ involving nodes that belong to blocks $g$ or $h$. 
More precisely,  with  $I_l=\{i: Z_i=l\}$ and  $V^\mathrm{merge}=\{(i,j): i\in I_g\cup I_h \text{ or }j\in I_g\cup I_h\}=\{(i,j): Z_i^\mathrm{merge}=h \text{ or }Z_j^\mathrm{merge}=h\}$,
$$A^\mathrm{merge}_{i,j}= \1\{\rho^\mathrm{merge}_{i,j}>0.5\},\quad \forall (i,j)\in V^\mathrm{merge}\qquad\text{and}\qquad A^\mathrm{merge}_{i,j}=A_{i,j},\quad  \forall (i,j)\notin V^\mathrm{merge},
$$
with $\rho^\mathrm{merge}_{i,j}=\P_{\hat \theta^\mathrm{merge}}(A_{i,j}=1\mid X, Z^\mathrm{merge})$, where $\hat \theta^\mathrm{merge}$ is the updated parameter estimate obtained with $Z^\mathrm{merge}$.
 	The expression of the variation $\Delta_{g\rightarrow h}$ is given in Proposition \ref{DeltaICLmerge}.
	
The variation $\Delta_{g \cup h}$ is evaluated for each pair of blocks $(g,h)$. Then the blocks $g$ and $h$, for which the variation is the largest, are definitely merged. 
Merging blocks is repeated  until no further merge   increases the $\mathrm{ICL}$ criterion anymore. A similar merge algorithm has also been proposed in \cite{Come15} for the inference in the classical SBM.  The algorithm is summarized in Algorithm S\ref{MergeICL}.

We introduce notations similar  to the previous algorithm  with superscript $^\mathrm{merge}$ to denote the quantities after merging blocks $g$ and $h$. For the number of nodes per block, we have 
$n_h^\mathrm{merge}=n_h+n_g$, $n_g^\mathrm{merge}=0$, and $n_q^\mathrm{merge}=n_q$ for all $q\notin\{g,h\}$.
For the set  of edges $(i,j)$ between blocks $q$ and $l$, 
$I^\mathrm{merge}_{g,l}=I^\mathrm{merge}_{l,g}=\emptyset$ for $l\in\dblbr Q$,
$I^\mathrm{merge}_{h,l}=I^\mathrm{merge}_{l,h}=\{(i,j):\cZ_{h,l}^{i,j\mathrm{merge}}A_{i,j}^\mathrm{merge}\} $ for $l\neq g$ and $I^\mathrm{merge}_{q,l}=I_{q,l}$ for $q,l\notin\{g,h\}$.
For the number of edges between blocks $q$ and $l$, we find
$n_{g,l}^\mathrm{merge}=n_{l,g}^\mathrm{merge}=0$ for $l\in\dblbr Q$, 
$n_{h,l}^\mathrm{merge}=n_{l,h}^\mathrm{merge}=|I^\mathrm{merge}_{h,l}|$ for $l\neq g$ and
$n_{q,l}^\mathrm{merge}=n_{q,l}$ for all other indices $(q,l)$.
The  number of possible edges between blocks $q$ and $l$ satisfies:
$m_{g,l}^\mathrm{merge}=m_{l,g}^\mathrm{merge}= 0$ for $l\in\dblbr Q$, 
$m_{h,l}^\mathrm{merge}=m_{l,h}^\mathrm{merge}= m_{h,l}+m_{g,l} $ for $l\neq \{g,h\}$ and
$m_{h,h}^\mathrm{merge} = m_{h,h}+m_{g,g}+m_{g,h}$. 
This implies that 
$\bar n_{g,l}^\mathrm{merge}=\bar n_{l,g}^\mathrm{merge}=0$ for $l\in\dblbr Q$, 
$\bar n_{h,l}^\mathrm{merge}=\bar n_{l,h}^\mathrm{merge}=m_{h,l}^\mathrm{merge} - n_{h,l}^\mathrm{merge}$ for $l\neq g$ and
$\bar n_{q,l}^\mathrm{merge}=\bar n_{q,l}$ for all other indices $(q,l)$.
Finally, we denote
$S_{h,l}^\mathrm{merge}$ the empirical variance associated with $\{X_{i,j}:(i,j)\in I^\mathrm{merge}_{h,l}\}$.
\begin{algorithm}[t]
	\SetAlgoLined
	\SetKwInOut{Input}{Input}
	\SetKwInOut{Output}{Output}
	\Input{Observation   $X$, clustering $Z$ into $Q$ clusters, parameter estimate $\theta$, edge probabilities $\rho$}
	\Output{Clustering $Z$, estimate $\theta$, edge probabilities $\rho$}
	 Compute $A$ with entries $A_{i,j}= \1\{\rho_{i,j}>0.5\}$\;
	\While{not converged}{
		\For{every pair of blocks $(g,h)$ with $g< h$}{
			Calculate $\Delta_{g \cup h}$ from  $X$, $Z$, $A$ and $A^\mathrm{merge}$ according to Proposition \ref{DeltaICLmerge}\;  
		}
		Merge the blocks $g$ and $h$ that achieve the largest increase   $\Delta_{g \cup h}$\;
		Update $Z, Q, \theta, \rho, A$ accordingly.
	}
	\caption{Merge ICL algorithm for NSBM inference} \label{MergeICL}
\end{algorithm}

\begin{proposition}  
	\label{DeltaICLmerge}
	The variation of the ICL when blocks $g$ and $h$ are merged is equal to 
	$$\Delta_{g \cup h} = \Delta^\mathrm{SBM}_{g \cup h} + \Delta^\mathrm{noise}_{g \cup h},$$	
	with
	\begin{align*}
&\Delta^\mathrm{SBM}_{g \cup h} 
=  \log  \frac{\Gamma(n_0 +n_h +n_g)\Gamma(Qn_0 + p)\Gamma(n_0)\Gamma((Q-1)n_0) }{\Gamma(n_0  +n_g)\Gamma(n_0 +n_h )  \Gamma((Q-1)n_0 + p)\Gamma(Qn_0)}     \\
&\quad + \sum_{l \in\dblbr Q\backslash  \{g\}}\log \frac{B(\eta_0+n_{h,l}^\mathrm{merge},\xi_0+\bar n_{h, l}^\mathrm{merge})}{B(\eta_{0}+n_{h,l}, \xi_0+\bar n_{h,l})} 
- \sum_{l\in\dblbr Q }  \log  B(\eta_0+n_{g,l},\xi_0+\bar n_{g,l}) 
+ Q B(\eta_{0},\xi_{0}), 
\end{align*}
and 
\begin{align*}
&\Delta^\mathrm{noise}_{g \cup h}  
= - \frac12 \sum_{(i,j)\in V^\mathrm{merge}} (A_{i,j}-A_{i,j}^\mathrm{merge})X_{i,j}^2
	- \left(Q+\sum_{ l\in\dblbr Q\backslash\{g\}} \left( n_{h,l}^\mathrm{merge}-n_{h,l}\right) - \sum_{ l\in\dblbr Q  } n_{g,l} \right)\log\sigma\\	
& -\frac{1}{2}  \sum_{ l\in\dblbr Q\backslash\{g\} } \log  \frac{\sigma^2 +\tau_0^2n_{h, l}^\mathrm{merge}}{\sigma^2 +\tau_0^2n_{h, l}} 
+\frac{1}{2}  \sum_{ l\in\dblbr Q  } \log(\sigma^2 +\tau_0^2n_{g, l})\\
& -\frac{\tau_0^2}{2\sigma^2}   \sum_{l\in\dblbr Q\backslash\{g\}} \left\{\frac{(n _{h, l}^\mathrm{merge})^2 S_{h, l}^\mathrm{merge}}{\tau_0^2  n_{h, l}^\mathrm{merge}+\sigma^2 } - \frac{n _{h, l}^2 S_{h, l}}{\tau_0^2 n_{h, l}+\sigma^2 } \right\}
 +\frac{\tau_0^2}{2\sigma^2}   \sum_{l\in\dblbr Q}   \frac{n _{g, l}^2 S_{g, l}}{\tau_0^2  n_{g, l}+\sigma^2} \\
&- \frac{1}{2}   \sum_{l\in\dblbr Q\backslash\{g\} }    \left\{\frac{1}{\sigma^2+\tau_0^2 n_{h, l}^\mathrm{merge}} \sum_{(i,j)\in I_{h, l}^\mathrm{merge}} (X_{i,j} - \rho_0)^2 - \frac{1}{\sigma^2+\tau_0^2 n_{h, l}} \sum_{(i,j)\in I_{h, l}} (X_{i,j} - \rho_0)^2\right\} \\
& +\frac{1}{2}   \sum_{l\in\dblbr Q }      \frac{1}{\sigma^2+\tau_0^2 n_{g, l}} \sum_{(i,j)\in I_{g, l}} (X_{i,j} - \rho_0)^2.
	\end{align*}
\end{proposition}

\section{NSBM with unknown variances}
The test statistics proposed in   \cite{Liu2013} are asymptotically normal, but their  limit variances are  unknown. To deal with this case, we propose a natural extension of our model and an adaptation of the algorithm, both presented in this section. 

\subsection{Model with unknown variances}
The definition of the NSBM is the same as in Section 2.1 of the main paper except the last layer, which describes the blurring mechanism.
It is replaced with 
 	$$
	(X_{i,j})_{(i,j)\in \cA}\:|\:Z,A \sim
	\bigotimes_{(i,j)\in \cA}   (1-A_{i,j}) \mathcal{N}(0,1) + A_{i,j}  \mathcal N(\mu_{Z_i,Z_j},\sigma^2_{Z_i,Z_j}),
	$$
with additional unknown symmetric parameter matrix $  \sigma^2= (\sigma_{q,l}^2)_{q,l}\in\R_+^{Q\times Q}$.
In this case the  unknown   model  parameter   is given  by $\theta=(\pi,w,\mu,   \sigma^2)$.
 
\subsection{ICL criterion}
To define the ICL criterion in the model with unknown variances, one has to choose an appropriate  prior for $\theta=(\pi,w,\mu,   \sigma^2)$.
As previously, we use a factorized prior, that is, 
$\boldsymbol\pi_\theta(\pi, w, \mu,\sigma^2) = 
\boldsymbol\pi_\pi(\pi)\boldsymbol \pi_w(w) \boldsymbol\pi_{\mu,\sigma^2}(\mu,\sigma^2) $, where  $\boldsymbol\pi_\pi$ and $\boldsymbol \pi_w$ are the same priors as before. 
For the Gaussian parameters, we assume $\boldsymbol\pi_{\mu,\sigma^2}(\mu,\sigma^2) =\otimes_{q\leq l}\boldsymbol\pi_{\mu_{q,l},\sigma_{q,l}^2}(\mu_{q,l},\sigma_{q,l}^2) $.
 In Bayesian statistics,  when both the mean and the variance of a Gaussian distribution are unknown, it is common to use a normal-inverse Gamma (NIG) distribution as prior, since it is conjugate. 
 That is, $\boldsymbol\pi_{\mu_{q,l},\sigma_{q,l}^2}(m,s)=\boldsymbol\pi_{\mu_{q,l}\mid \sigma_{q,l}^2}(m\mid s)\boldsymbol\pi_{\sigma_{q,l}^2}(s)$, where $\boldsymbol\pi_{\sigma_{q,l}^2}$ is the inverse Gamma distribution $\mathrm{IG}(c_0,d_0)$ with hyperparameters $c_0,d_0$, that is, the prior density is given by
 $$
 \boldsymbol\pi_{\sigma_{q,l}^2}(z)=
 \frac{d_0^{c_0}}{\Gamma(c_0)} z^{-c_0-1}e^{-\frac{d_0}{z}},\quad  z>0,
 $$
 and 
 $\boldsymbol\pi_{\mu_{q,l}\mid \sigma_{q,l}^2}(\cdot\mid s)$ 
 is the Gaussian distribution $\cN\left(a_0,\frac{s}{b_0}\right)$ with hyperparameters $a_0, b_0$.
 We use the standard values for the  hyperparameters, which are $a_0=0$ and $b_0=c_0=d_0=1$.  

With this choice of the prior, the ICL has an analytical form stated in the following proposition.

\begin{proposition}  
	\label{ICL_NIG}
	The integrated complete data log-likelihood is given by
	$$\mathrm{ICL^{NSBM}}(X,A,Z) =   \log p(X, A, Z) =  \mathrm{ICL^{SBM}}(A,Z) + \mathrm{ICL^{noise}}(X,A,Z),$$
	with $ \mathrm{ICL^{SBM}}(A,Z)$ given by (2) in the main paper 
	and
	\begin{align*}
		\mathrm{ICL^{noise}}(X,A,Z) =&\log p(X | A, Z) 	\notag\\
		=& - \frac{N}{2}  \log(2\pi)
		 - N_Q\log\left(\frac{\Gamma(c_0)}{d_0^{c_0}\sqrt{b_0}}\right) - \frac{1}{2}\sum_{(i,j) \in \mathcal{A}} (1-A_{i,j})  X_{i,j}^2\\
		&+ \sum_{q \le l}\left\{ \log \Gamma\left(c_0+\frac{n_{q,l}}2\right)  -\frac12\log(b_0+n_{q,l})  -  
		\left(c_0+\frac{n_{q,l}}2\right) \log(d_{q,l})\right\},   \notag 
	\end{align*}
 with the notations of Proposition 1 
 in the main paper and	 
 $$
 d_{q,l}= d_0+ \frac{n_{q,l}S_{q,l}}{2}+\frac{n_{q,l}b_0}{2(b_0+n_{q,l})}(\overline{X}_{q,l}-a_0)^2,
 $$
 where $ \overline{X}_{q,l}$ denotes the sample mean of the observations $X_{i,j}$ with $(i,j)\in I_{q,l}$. 
\end{proposition}

 \begin{proof}[Proof of Proposition \ref{ICL_NIG}]
Compared to the previous case, only  the term $\mathrm{ICL^{noise}} $ changes.
We have
\begin{align*}
	\mathrm{ICL^{noise}} (X,A,Z) &= \log \int_{\R_+^{N_Q}} \int_{ \R^{N_Q}} p(X|A, Z, { \mu}, \sigma^2 )\boldsymbol \pi({ \mu, \sigma^2})d{ \sigma^2}d{ \mu} \\
	&= \sum_{(i,j)\in\mathcal{A}} (1-A_{i,j}) \log (f_{\mathcal{N}(0,1)}(X_{i,j})) + \sum_{q\leq l } \log (\mathcal{J}_{q,l}),
\end{align*}
with, for all ${q\leq l }$,
$$
\mathcal{J}_{q,l} 
=  \int_{\R_+} \int_{ \R}   \left[\prod_{(i,j) \in I_{q,l}}f_{\mathcal{N}(\mu_{q,l},\sigma_{q,l}^2)}(X_{i,j}) \right]f_{\mathcal{N}(a_0, \sigma_{q,l}^2/b_0) }(\mu_{q,l}) f_{IG(c_0,d_0)}(\sigma_{q,l}^2) d\mu_{q,l} d\sigma_{q,l}^2.
$$
By  dropping subscripts $_{q,l}$ for readability, we obtain 
\begin{align*}
	\mathcal{J}_{q,l} 
	&=(2\pi)^{-\frac n2-\frac 12} \frac{\sqrt{b_0}d_0^{c_0}}{\Gamma(c_0)} \\
&\qquad\times \int_{\R_+} \int_{ \R} (\sigma^2)^{-(c_0+\frac{n}{2}+\frac{3}{2})}e^{-\frac{d_0}{\sigma^2}} \exp\left(-\frac{1}{2\sigma^2}\left(\sum_{  I}(X_{i,j}-\mu)^2+b_0(\mu-a_0)^2\right)\right)d\mu d\sigma^2\\
	 &= (2\pi)^{-\frac n2-\frac 12} \frac{\sqrt{b_0}d_0^{c_0}}{\Gamma(c_0)} \int_{\R_+} (\sigma^2)^{-c_0- \frac{n}2-\frac{3}{2}}e^{ -\frac{d }{\sigma^2}} 
	  \int_{ \R}   \exp\left\{-\frac{(b_0+n) \left(\mu-\bar a \right)^2}{2\sigma^2}\right\}d\mu d\sigma^2\\
	 &= (2\pi)^{-\frac n2-\frac 12} \frac{\sqrt{b_0}d_0^{c_0}}{\Gamma(c_0)} \int_{\R_+} (\sigma^2)^{-c_0- \frac{n}2-\frac{3}{2}}e^{ -\frac{d }{\sigma^2}} 
	 \sqrt{\frac{2\pi\sigma^2}{b_0+n}} \int_{ \R} f_{\mathcal N(\bar a, \frac{\sigma^2}{b_0+n})}(\mu) d\mu d\sigma^2\\
	 &= (2\pi)^{-\frac n2} \sqrt{\frac{b_0}{b_0+n}} \frac{d_0^{c_0}}{\Gamma(c_0)} \frac{\Gamma(c_0+\frac n2)}{d^{c_0+\frac n2}}\int_{\R_+} f_{IG(c_0+\frac n2, d)}(\sigma^2)
   d\sigma^2\\
	&=(2\pi)^{-\frac n2}\sqrt{\frac{b_0}{b_0+n}}\frac{\Gamma(c_0+\frac n2)}{\Gamma(c_0)}\frac{d_0^{c_0}}{d ^{c_0+\frac n2}},
		\end{align*}
		with
$\bar a = (n_{q,l}\overline{X}_{q,l}+a_0b_0)/(b_0+n_{q,l}).$ 
Moreover,
\begin{align*}
 \sum_{(i,j)\in\mathcal{A}} (1-A_{i,j}) \log (f_{\mathcal{N}(0,1)}(X_{i,j}))
 &=-\frac 12\log(2\pi) \sum_{(i,j)\in\mathcal{A}} (1-A_{i,j}) -\frac12 \sum_{(i,j)\in\mathcal{A}} (1-A_{i,j})X_{i,j}^2.
\end{align*}
Putting all terms together and as $N_Q=\frac{Q(Q+1)}2$ and  $N=\sum_{q\leq l } n_{q,l}+ \sum_{(i,j)\in\mathcal{A}} (1-A_{i,j})$,
the result  follows.
\end{proof}

\subsection{Greedy ICL inference algorithm}
To maximize the ICL criterion in the NSBM with unknown variances, Algorithm 1 from the main paper can still be used, but  the expression of the variation $\Delta_{g \rightarrow h} $ of the ICL when node $i^*$ is swapped from block $g$ to block $h$ is slightly changes. The decomposition 
\begin{equation*}
 			\Delta_{g \rightarrow h} = \Delta^\mathrm{SBM}_{g \rightarrow h}+\Delta^\mathrm{noise}_{g \rightarrow h},\end{equation*}
still holds, where $\Delta^\mathrm{SBM}_{g \rightarrow h}$ is given by Proposition \ref{DeltaICLswap}. Also the changes in the count statistics $n_q, n_{q,l}$ etc. are the same. The new expression of $\Delta^\mathrm{noise}_{g \rightarrow h}$ is given in the following proposition.

\begin{proposition}  
 		\label{DeltaICLswap_NIG}
According to two situations, the following expressions hold.
\paragraph{Case 1} If $n_g>1$, that is $i^*$ is not the only node assigned to block $g$, then 
 		\begin{align*}
 			\Delta^\mathrm{noise}_{g \rightarrow h} 
			=  &-\frac12\sum_{j \in\dblbr p} (A_{i^*,j} - A^\mathrm{swap}_{i^*,j})X_{i^*,j}^2 \\
			&+\sum_{(q,l)\in\mathcal S_{q,h}}\left\{\log\frac{\Gamma\left(c_0+\frac{n_{q,l}^\mathrm{swap}}2\right)}{\Gamma\left(c_0+\frac{n_{q,l}}2\right)} -\frac12\log\frac{b_0+n_{q,l}^\mathrm{swap}}{b_0+n_{q,l}}
			-\log\frac{(d_{q,l}^\mathrm{swap})^{c_0+n_{q,l}^\mathrm{swap}/2}}{(d_{q,l})^{c_0+n_{q,l}/2}}
			\right\}.
 		\end{align*}
 		
 		\paragraph{Case 2}  If $n_g=1$, that is, after the swap block $g$  is empty, so that block $g$ disappears from the model  and the number of blocks diminishes by 1, that is $Q^\mathrm{swap} = Q-1$, then 
 		\begin{align*}
 			\Delta^\mathrm{noise}_{g \rightarrow h}= &
			Q\log \frac{\Gamma(c_0)}{d_0^{c_0}\sqrt{b_0}} 
			-\frac12\sum_{j \in\dblbr p} (A_{i^*,j} - A^\mathrm{swap}_{i^*,j})X_{i^*,j}^2 \\
			&+\sum_{l\in\dblbr Q\backslash\{g\}}\left\{\log\frac{\Gamma\left(c_0+\frac{n_{h,l}^\mathrm{swap}}2\right)}{\Gamma\left(c_0+\frac{n_{h,l}}2\right)} -\frac12\log\frac{b_0+n_{h,l}^\mathrm{swap}}{b_0+n_{h,l}}
			-\log\frac{(d_{h,l}^\mathrm{swap})^{c_0+n_{h,l}^\mathrm{swap}/2}}{(d_{h,l})^{c_0+n_{h,l}/2}}
			\right\}\\
		&-	\sum_{l\in\dblbr Q}\left\{ \log \Gamma\left(c_0+\frac{n_{g,l}}2\right)  -\frac12\log(b_0+n_{g,l})  -  
		\left(c_0+\frac{n_{g,l}}2\right) \log(d_{g,l})\right\}.
 		\end{align*}
 	\end{proposition}

Finally,   Algorithm \ref{MergeICL} to merge blocks can be applied where only the expression of the term $\Delta^\mathrm{noise}_{g \cup h}$ is changed. 

\begin{proposition} 
	\label{DeltaICLmerge_NIG}
When all nodes from block $g$ are moved to block $h$, then it holds
	\begin{align*}
		\Delta^\mathrm{noise}_{g \cup h}  
		=  & Q\log \frac{\Gamma(c_0)}{d_0^{c_0}\sqrt{b_0}} 
		-\frac12\sum_{(i,j)\in V^\mathrm{merge}} (A_{i,j} - A^\mathrm{merge}_{i,j})X_{i,j}^2 \\
					&+\sum_{l\in\dblbr Q\backslash\{g\}}\left\{\log\frac{\Gamma\left(c_0+\frac{n_{h,l}^\mathrm{merge}}2\right)}{\Gamma\left(c_0+\frac{n_{h,l}}2\right)} -\frac12\log\frac{b_0+n_{h,l}^\mathrm{merge}}{b_0+n_{h,l}}
			-\log\frac{(d_{h,l}^\mathrm{merge})^{c_0+n_{h,l}^\mathrm{merge}/2}}{(d_{h,l})^{c_0+n_{h,l}/2}}
			\right\}\\
		&-	\sum_{l\in\dblbr Q}\left\{ \log \Gamma\left(c_0+\frac{n_{g,l}}2\right)  -\frac12\log(b_0+n_{g,l})  -  
		\left(c_0+\frac{n_{g,l}}2\right) \log(d_{g,l})\right\}.
	\end{align*}
\end{proposition}

\section{Numerical results}

\begin{figure}[t]
	\includegraphics[scale=1]{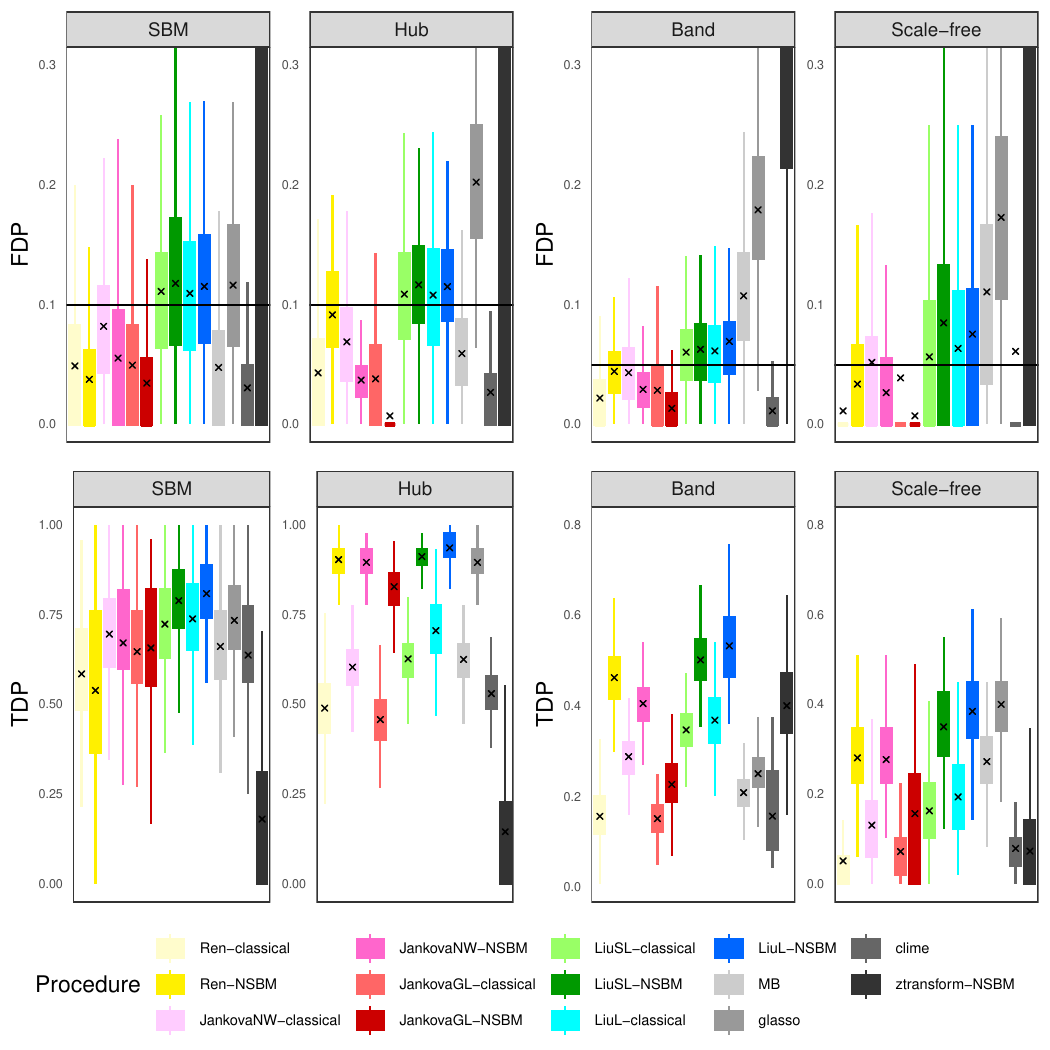}	
	\centering
	\caption{Boxplots of the FDP and  TDP for all procedures described in Section 4.2  of the main paper 
	for the settings SBM, hub, band and scale-free for 200 simulated datasets, where $p=50$ and $n=100$. Horizontal lines  represent the nominal level $\alpha$, with $\alpha=0.1$   in  the SBM and the hub setting 
and $\alpha=0.05$ in the band and scale-free cases.
  The crosses in the boxplots correspond to the sample FDR  (resp. TDR) defined as the mean of the FDP (resp.  TDP).
} 
	\label{fig:simp50}
\end{figure}

\begin{figure}[t]
	\includegraphics[scale=1]{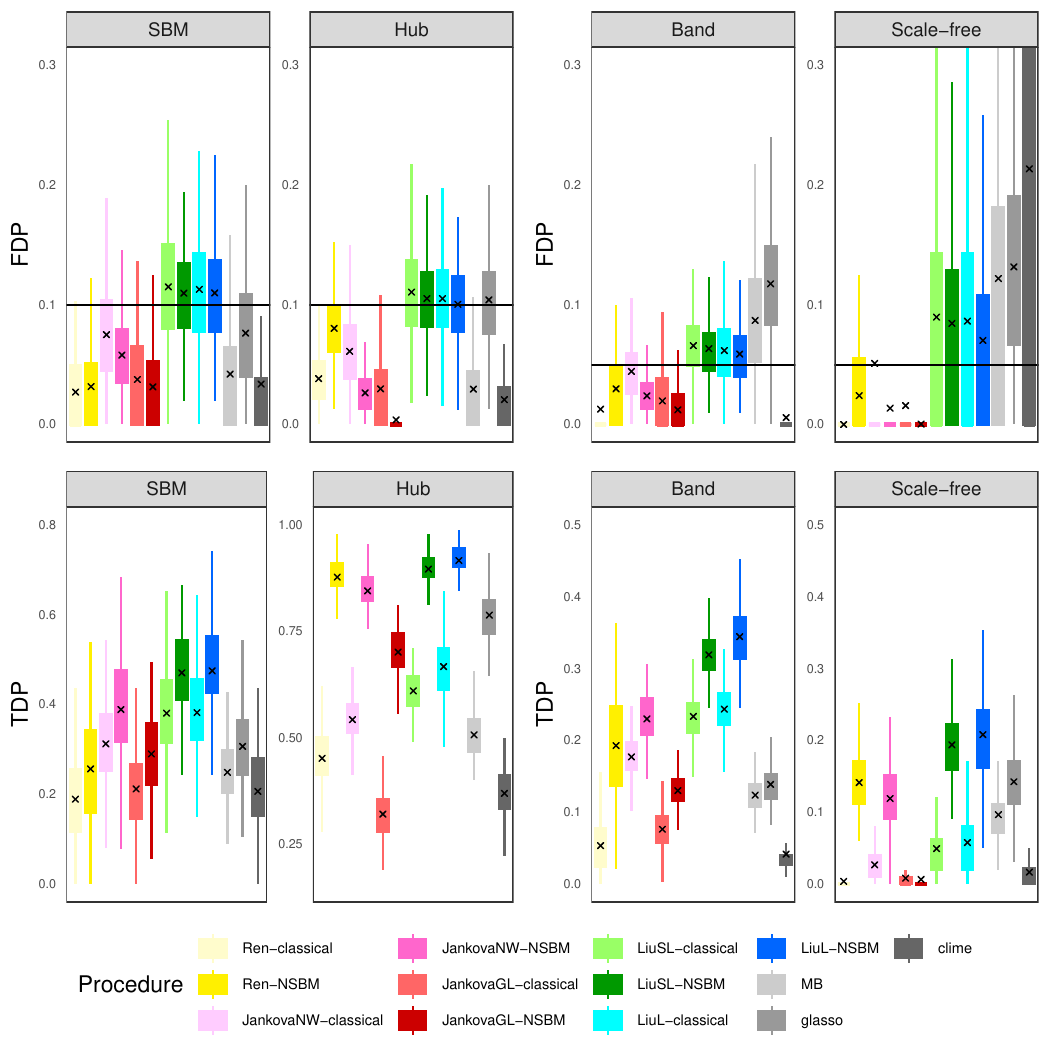}	
	\centering
	\caption{Boxplots of the FDP and  TDP for all procedures described in Section 4.2  of the main paper 
	for the settings SBM, hub, band and scale-free for 200 simulated datasets, where $p=100$ and $n=100$. Horizontal lines  represent the nominal level $\alpha$, with $\alpha=0.1$   in  the SBM and the hub setting 
and $\alpha=0.05$ in the band and scale-free cases.
  The crosses in the boxplots correspond to the sample FDR  (resp. TDR) defined as the mean of the FDP (resp.  TDP).
} 
	\label{fig:simp100}
\end{figure}

To complete the numerical study of Section 5.1 in the main paper,   this section provides   results for settings with 
    $p<n$  (Figure \ref{fig:simp50}) and  $p=n$   (Figure \ref{fig:simp100}).    The \texttt{ztransform-NSBM} procedure
    corresponds to the NSBM approach applied to test statistics computed from the inverse of the sample covariance matrix \citep{Anderson} and is only applicable when $n$ is larger than $p$.
    The conclusions are roughly the same as in the more difficult setting with $p>n$. Namely, in general we observe a significant improvement by applying the NSBM approach over the classical multiple testing approach for all test statistics. 

\bibliographystyle{Chicago}

\bibliography{biblio}